\documentclass[12pt]{article}
\usepackage{amsmath}
\usepackage{amssymb}
\usepackage{amsthm}
\usepackage{graphicx}
\usepackage{lineno}
\usepackage{float}
\usepackage[all]{xy}
\usepackage{mathrsfs}

\usepackage{xcolor}
\usepackage{tikz}
\usetikzlibrary{tqft}

\newtheorem{theorem}{Theorem}
\newtheorem{lemma}{Lemma}

\newtheorem{definition}{Definition}
\newtheorem{proposition}{Proposition}

\newcommand{\ad}{{\rm ad}}

\newcommand{\Det}{{\rm Det}}

\newcommand{\Cl}{{\rm Cl}}
\newcommand{\ch}{{\rm ch}}

\newcommand{\End}{{\rm End}}

\newcommand{\Hom}{{\rm Hom}}

\newcommand{\Index}{{\rm Ind}}

\newcommand{\GL}{{\rm GL}}

\newcommand{\SL}{{\rm SL}}
\newcommand{\SO}{{\rm SO}}
\newcommand{\SU}{{\rm SU}}

\newcommand{\Spin}{{\rm Spin}}

\newcommand{\A}{\boldsymbol{\mathcal{A}}}

\newcommand{\mfF}{\mathfrak{F}}

\newcommand{\mfM}{\mathfrak{M}}
\newcommand{\mfO}{\mathfrak{O}}
\newcommand{\mfP}{\mathfrak{P}}

\newcommand{\mfp}{\mathfrak{p}}

\newcommand{\mbA}{\mathbf{A}}

\newcommand{\mbC}{\mathbf{C}}

\newcommand{\lra}{\longrightarrow}

\begin{document}

\title{\bf Sequential measurements, TQFTs, \\
and TQNNs}

\author{{Chris Fields$^a$, James F. Glazebrook$^{b,c}$ and Antonino Marcian\`{o}$^{d,e}$}\\ \\
{\it$^a$ 23 Rue des Lavandi\`{e}res, 11160 Caunes Minervois, FRANCE}\\
{fieldsres@gmail.com}\\
{ORCID: 0000-0002-4812-0744}\\
{\it$^b$ Department of Mathematics and Computer Science,} \\
{\it Eastern Illinois University, Charleston, IL 61920 USA} \\
{\it$^c$ Adjunct Faculty, Department of Mathematics,}\\
{\it University of Illinois at Urbana-Champaign, Urbana, IL 61801 USA}\\
{jfglazebrook@eiu.edu}\\
{ORCID: 0000-0001-8335-221X}\\
{\it$^d$ Center for Field Theory and Particle Physics \& Department of Physics} \\
{\it Fudan University, Shanghai, CHINA} \\
{marciano@fudan.edu.cn} \\
{\it$^e$ Laboratori Nazionali di Frascati INFN, Frascati (Rome), ITALY} \\
{ORCID: 0000-0003-4719-110X}
}

\maketitle

{\bf Abstract} \\
We introduce novel methods for implementing generic quantum information within a scale-free architecture. For a given observable system, we show how observational outcomes are taken to be finite bit strings induced by measurement operators derived from a holographic screen bounding the system. In this framework, measurements of identified systems with respect to defined reference frames are represented by semantically-regulated information flows through distributed systems of finite sets of binary-valued Barwise-Seligman classifiers. Specifically, we construct a functor from the category of cone-cocone diagrams (CCCDs) over finite sets of classifiers, to the category of finite cobordisms of Hilbert spaces. We show that finite CCCDs provide a generic representation of finite quantum reference frames (QRFs). Hence the constructed functor shows how sequential finite measurements can induce TQFTs. The only requirement is that each measurement in a sequence, by itself, satisfies Bayesian coherence, hence that the probabilities it assigns satisfy the Kolmogorov axioms.  We extend the analysis so develop topological quantum neural networks (TQNNs), which enable machine learning with functorial evolution of quantum neural 2-complexes (TQN2Cs) governed by TQFTs amplitudes, and resort to the Atiyah-Singer theorems in order to classify topological data processed by TQN2Cs. We then comment about the quiver representation of CCCDs and generalized spin-networks, a basis of the Hilbert spaces of both TQNNs and TQFTs. We finally review potential implementations of this framework in solid state physics and suggest applications to quantum simulation and biological information processing.

\tableofcontents

\section{Introduction}

Let $U$ be an isolated, finite-dimensional quantum system and consider a fixed bipartite decomposition $U = AB$ that induces a fixed interaction $H_{AB} = H_U - (H_A + H_B)$, and assume that $H_{AB}$ is weak enough that $|AB(t) \rangle$ can be considered separable, i.e. $|AB(t) \rangle = |A(t) \rangle |B(t) \rangle$ over the entire time period of interest.  In this case, we can choose basis vectors $|i^k \rangle$ so that:

\begin{equation} \label{ham}
H_{AB} = \beta_k K_B\, T_k \sum_i^N \alpha^k_i M^k_i,
\end{equation}
\noindent
where $K_B$ stands for the Boltzmann constant, $T$ for the absolute temperature of the environment, $k =~A$ or $B$, the $M^k_i$ are $N$ Hermitian operators with eigenvalues in $\{ -1,1 \}$, the $\alpha^k_i \in [0,1]$ are such that $\sum^N_i \alpha^k_i = 1$, and $\beta_k \geq \ln 2$ is an inverse measure of $k$'s thermodynamic efficiency that depends on the internal dynamics $H_k$ --- see e.g. Refs.~\cite{fm:19, fm:20, fg:20, fgm:21, addazi:21}.  For fixed $k$, the operators $M^k_i$ clearly must commute, i.e. $[M^k_i, M^k_j] = M^k_i M^k_j - M^k_j M^k_i = 0$ for all $i, j$; hence $H_{AB}$ is swap-symmetric under the permutation group $S_N$ for each $k$.  We can, therefore, write $N = \dim(H_{AB})$, i.e. the eigenvalues $H_{AB}$ can be encoded by $2^N$ distinct $N$-bit strings.  Via the Holographic Principle \cite{hooft:83, susskind:95, bousso:02, almheiri:21}, we can consider $H_{AB}$ to be defined at a finite boundary $\mathscr{B}$ that encodes, at each instant, one of these $2^N$ bit strings \cite{fm:20, fgm:21, addazi:21}.  We assume a discrete topology for $\mathscr{B}$, with each point encoding one bit; in a widely accepted geometric picture, a ``point'' on $\mathscr{B}$ corresponds to a 4$L_P^2$ pixel, with $L_P$ denoting the Planck length.

As discussed in \cite{fm:20, fgm:21, addazi:21}, we can regard $\mathscr{B}$ as implemented by an array $q_i$ of $N$ mutually-noninteracting qubits, and regard each of the operators $M^k_i$ as implemented by a $z$-spin operator $s_{z(k,i)}$, where the notation $z(k,i)$ indicates that $A$ and $B$ each have ``free choice'' of $z$ direction for each qubit $q_i$.  In this case, we can associate a $2N$-dimensional space $\mathcal{H} = \mathcal{H}_A \bigotimes\mathcal{H}_B$ with $\mathscr{B}$; the components $\mathcal{H}_A$ and $\mathcal{H}_B$ of $\mathcal{H}$ are $N$-dimensional Hilbert spaces spanned by basis vectors $\{ z(A,i) \}$ and $\{ z(B,i) \}$, respectively.  In this representation, the operators $M^k_i$ act on the Hilbert space $\mathcal{H}_k$, with $k =~A$ or $B$ as in Eq. \eqref{ham}.  Permutations $q_i \rightarrow q_j$ of qubits on $\mathscr{B}$ correspond to permutations of pairs $z(A,i), z(B,i) \rightarrow z(A,j), z(B,j)$ of basis vectors in $\mathcal{H}$.  These correspond in turn to permutations $M^A_i, M^B_i \rightarrow M^A_j, M^B_j$ of pairs of measurement operators and hence to permutations of pairs of bit strings encoding the current eigenvalue of $H_{AB}$; such permutations effectively reset the zero point of the energy scale and so have no physical consequence, consistent with the $S_N$ symmetry of $H_{AB}$.

We now adopt the perspective of system $A$ as an ``observer'' and consider observational outcomes, i.e. bit strings in $\{ 0,1 \}^N$, obtained by deploying the $M^A_i$ on $\mathcal{H}_A$, i.e. on $\mathscr{B}$ viewed as the boundary of $A$.  Given the above, we can represent, from $A$'s perspective, the action of the evolution operator $\mathcal{P}_U(t) = \exp(-(\imath/\hbar)H_U(t))$ from $t = i \rightarrow t = k$ by requiring the following diagram to commute:

\begin{equation} \label{prop}
\begin{gathered}
\begin{tikzpicture}
\node at (0,2) {$\mathscr{B}$};
\node at (-0.6,1) {$f_{\mathscr{B}}(i)$};
\draw [thick, ->] (0,0.3) -- (0,1.7);
\node at (0.1,0) {$\{ 0,1 \}^N$};
\draw [thick, ->] (0,-0.3) -- (0,-1.7);
\node at (-0.7,-1) {$f_{\mathcal{H}_A}(i)$};
\node at (0,-2) {$\mathcal{H}_A$};
\node at (2,2.3) {$\mathrm{Id}$};
\draw [thick, ->] (0.4,2) -- (3.7,2);
\node at (2,1.3) {$\mathcal{P}_U(t)$};
\draw [thick, ->] (0.2,1) -- (3.8,1);
\node at (2,0.3) {$\mathrm{Id}$};
\draw [thick, ->] (0.8,0) -- (3.3,0);
\node at (2,-0.7) {$\mathcal{P}_U(t)$};
\draw [thick, ->] (0.2,-1) -- (3.8,-1);
\node at (2,-1.7) {$\mathrm{Id}$};
\draw [thick, ->] (0.4,-2) -- (3.7,-2);
\node at (4,2) {$\mathscr{B}$};
\node at (4.6,1) {$f_{\mathscr{B}}(k)$};
\draw [thick, ->] (4,0.3) -- (4,1.7);
\node at (4.1,0) {$\{ 0,1 \}^N$};
\draw [thick, ->] (4,-0.3) -- (4,-1.7);
\node at (4.7,-1) {$f_{\mathcal{H}_A}(k)$};
\node at (4,-2) {$\mathcal{H}_A$};
\end{tikzpicture}
\end{gathered}
\end{equation}
\noindent
where the maps $f_{\mathscr{B}}$ and $f_{\mathcal{H}_A}$ encode $N$-bit strings as eigenvalues of $H_{AB}$ onto $\mathscr{B}$ and as normalized vectors into $\mathcal{H}_A$, respectively, and Id indicates the (time- and encoding-independent) Identity maps on $\mathscr{B}$, $\{ 0,1 \}^N$, and $\mathcal{H}$.  We can, in other words, consider all the time dependence of $\mathcal{P}_U(t)$ to be in the encoding functions $f_{\mathscr{B}}$ and $f_{\mathcal{H}}$.  Hence we can consider $\mathcal{P}_U(t)$ as implementing a (quantum) computation on (encoded) bit strings in $\{ 0,1 \}^N$.

We will be interested in $A$'s sequential observations of a sector $S$ of $\mathscr{B}$ with dimension $\dim(S) = n \ll N$.  We will write the subset of operators acting on $S$ as $\{ M^S_j \} \subsetneq \{ M^A_i \}$.  From $A$'s perspective, $S$ is an (apparent, $A$-relative) ``system'' with an (effective) Hilbert space $\mathcal{H}_S$ with $\dim(\mathcal{H}_S) = n$ on which the operators $M^S_j$ act.

For $A$'s sequential observations of $S$ to be physically meaningful, they must be comparable.  If the intervals between observations are short, comparability requires that two conditions are met: 1) $S$ must have some ``reference'' component $R$ with a state $|R \rangle$ or sampled density $\rho_R$ that remains invariant for the entire duration of observations, and 2) measurements of the time varying state $|P(t) \rangle$ or sampled density $\rho_P(t)$ of the ``pointer'' component $P$ of $S$ must be assigned well-defined units as part of the measurement process.  In this case, we have an effective Hilbert space decomposition $\mathcal{H}_S = \mathcal{H}_{R} \bigotimes \mathcal{H}_{P}$.  As discussed in \cite{fm:20, fg:20, fgm:21}, these requirements for comparability of sequential measurements are met if $A$ deploys a quantum reference frame (QRF) \cite{aharonov:84, bartlett:07} $\mathbf{S} = \mathbf{RP}$, where the component $\mathbf{R}$ identifies $S$ by measuring the invariant state $|R \rangle$ or sampled density $\rho_R$ and the component $\mathbf{P}$ measures the time-varying state $|P(t) \rangle$ or sampled density $\rho_P(t)$ and assigns the relevant units.  Under these conditions the QRFs $\mathbf{R}$ and $\mathbf{P}$ clearly must commute and the joint state $|S(t) \rangle = |RP(t) \rangle$ or sampled density $\rho_S(t) = \rho_{RP}(t)$ must be separable, i.e. $|S(t) \rangle = |R \rangle |P(t) \rangle$ or $\rho_S(t) = \rho_R \rho_P(t)$, respectively.  Note that ``systems'' $S$, $R$, and $P$ are defined, and hence this separability requirement is stated, strictly relative to $A$.  We can in this case simply define $S =_{def} \text{dom}(\mathbf{S})$, $R =_{def} \text{dom}(\mathbf{R})$, $P =_{def} \text{dom}(\mathbf{P})$ and require $S = RP$.

We show in what follows that any such sequential observations induce, with no further assumptions, a topological quantum field theory (TQFT) of the observed sector $S$.  We do this by establishing three constructive results. First, we construct in \S \ref{QRF2CCCD} an explicit representation of QRFs as category-theoretic structures, specifically cone-cocone diagrams (CCCDs) with Barwise-Seligman \cite{barwise:97} classifiers as nodes and morphisms (``infomorphisms'') between them as directed links \cite{fg:19a}.  This representation generalizes hierarchical Bayesian networks, standard multi-layer artificial neural networks (ANNs), and variational auto-encoders (VAEs).  It has the advantage of making quantum noncontextuality manifest as a diagram commutativity requirement \cite{fg:22}.  Second, we construct in \S \ref{CCCD-cat} a category $\mathbf{CCCD}$ in which the objects are CCCDs representing QRFs, and show how morphisms in this category represent sequential measurements.  Third, we construct in \S \ref{CCCD2cobord} a functor $\mathfrak{F}: \mathbf{CCCD} \rightarrow \mathbf{Cobord}$, the category of cobordisms.  As a TQFT can be defined as a functor $\mathbf{Cobord} \rightarrow \mathbf{Hilb}$, the category of (finite-dimensional) Hilbert spaces \cite{atiyah:88} (see also e.g. \cite{quinn:95}), we have that sequential observations induce TQFTs, in particular, sequential observations of $S$ induce TQFTs with copies of the Hilbert space $\mathcal{H}_S$ as boundaries.
In \S \ref{CCCD-TQNN} we establish a link between sequential measurements and topological quantum neural networks (TQNNs), a framework that encode deep neural networks (DNNs) in their semi-classical limit, moving from the relation between TQFT and TQNN constructed in \cite{marciano:20}. We then show how topological data processed by the functorial evolution of TQNNs, encoded by topological quantum neural 2-complexes (TQN2Cs), can be classified resorting to the Atiyah-Singer theorems, when spin manifolds are taken into account. In \S \ref{app} we discuss several applications of this framework, from gauge and effective field theories to quiver representations of TQFTs. We draw analogies with some quantum gravity models, including the spin-foam representation of functors over the Hilbert space of discretised geometries and quiver representations of generalized spin-network states in noncommutative geometry. We review potential instantiations of these theoretical structures in notable systems of state solid physics, including topological insulators, phases of matter described by string-nets and trapped ions and atoms. Finally, in \S \ref{con} we spell out our conclusions, prospect future developments and illustrate possible implementations in topological quantum computation and biological information processing.
Accordingly, much of what is accounted for here is in line with components of the compelling program as documented in \cite{chen:17}.

Before proceeding, we mention two illustrative examples.  First, suppose that $A$ observes $S$ at $A$-relative time $t_A = i$, and then later, at some $t_A = k$ observes distinct, non-overlapping (topologically disconnected) systems/sectors $S_1$ and $S_2$, and infers that $S$ underwent a ``fission'' process, denoted by $S \rightarrow S_1 S_2$, corresponding to a Hilbert-space decomposition $\mathcal{H}_S = \mathcal{H}_{S_1} \bigotimes \mathcal{H}_{S_2}$.  This fission/decomposition process may leave the joint system in an entangled state, i.e. $|S_1 S_2 \rangle \neq |S_1 \rangle |S_2 \rangle$.  Parametric down-conversion $\gamma \mapsto \gamma_1, \gamma_2$ of a photon to an entangled photon pair provides an example of this type.  The measurements made in such an experiment --- particularly the measurements of the reference sector $R$ --- clearly must support the inference of a continuous process, and hence a TQFT, mapping $S \rightarrow S_1 \times S_2$.

As a second example, consider an apparatus $S$ with ``settings'' that allow measuring alternative pointer components $P$ and $Q$, e.g. spin orientations $s_z$ and $s_x$, for which the relevant QRFs $\mathbf{P}$ and $\mathbf{Q}$ do not commute.  In this case, the effective Hilbert space decomposition is $\mathcal{H}_S = \mathcal{H}_{R} \otimes \mathcal{H}_{P} = \mathcal{H}_{R} \otimes \mathcal{H}_{Q}$, i.e. $\mathcal{H}_P = \mathcal{H}_Q$.  The measurements made must support the inference of a continuous process, and hence a TQFT, mapping $RP \rightarrow RQ$, i.e. implementing the basis rotation that distinguishes $P$ from $Q$.

\section{Representing QRFs as CCCDs} \label{QRF2CCCD}

\subsection{QRFs as computations} \label{QRF-comp}

John Wheeler famously characterized experiments as ``questions to Nature'' with yes/no answers \cite{wheeler:83}.  Taking this remark literally allows a completely general characterization of QRFs as physically implemented computations.  To begin, note that any QRF plays complementary roles in measuring and preparing quantum states.  The 10 cm mark on a meter stick, for example, can both report that a rod is 10 cm long and guide a cut that renders it 10 cm long; the $\uparrow$ direction on a polarizing filter can, similarly, both measure and prepare polarized light.  As the input to any finite QRF can be represented as a finite bit string per the discussion above, we can represent the action of any QRF as:

\begin{equation} \label{QRF1}
<\mbox{bit-string}> ~ \xrightarrow{Measure} ~\mathrm{outcome ~value}~ \xrightarrow{Prepare} ~ <\mbox{bit-string}>^\prime.
\end{equation}
\noindent
No generality is lost, moreover, by regarding any finite QRF as a collection of $m$ ``elementary'' or ``single pointer position'' QRFs, each of which reports or prepares a single value, i.e. each of which implements one yes/no question.  For simplicity, we assume these questions are orthogonal, i.e. that elementary QRFs implement projective measurements. A meter stick, for example, can be regarded as a collection of pointer positions, each indicated by a tick mark.  Each tick mark acts as an elementary QRF that reports the outcome value ``1'' for exactly one length (up to some finite resolution that determines $m$, e.g. $\pm 1$ mm for $m = 1,001$) and reports the outcome value ``0'' for all other lengths.  The tick mark at 10 cm, for example, reports ``1'' for lengths of $10.0 \pm 0.1$ cm and ``0'' for all other lengths, and can be used to either measure or prepare 10 cm lengths .  We can represent the action of such an elementary QRF as:

\begin{equation} \label{QRF2}
\{0, 1 \}^m \xrightarrow{\{\overrightarrow{g_i} \}} ~\{0, 1\}~ \xrightarrow{\{\overleftarrow{g_i} \}} \{0, 1 \}^m,
\end{equation}
\noindent
where the $\overrightarrow{g_i}$ implement measurement and the $\overleftarrow{g_i}$ implement preparation, and such that the following conditions hold:

\begin{equation} \label{QRF3a}
\exists ! e \in \{0, 1 \}^m, \overrightarrow{g} \in \{\overrightarrow{g_i} \}, ~\mathrm{and}~ \overleftarrow{g} \in \{\overleftarrow{g_i} \}
 ~\mathrm{such ~that}~ \overrightarrow{g} : e \mapsto 1 ~\mathrm{and}~ \overleftarrow{g}: 1 \mapsto e,
\end{equation}
\noindent
and

\begin{equation} \label{QRF3b}
\begin{aligned}
&\forall e^{\prime} \neq e \in \{0, 1 \}^m, ~ \exists ! \overrightarrow{g^{\prime}} \neq \overrightarrow{g} \in \{\overrightarrow{g_i} \}, ~\mathrm{and}~ \overleftarrow{g^{\prime}} \neq \overleftarrow{g} \in \{\overleftarrow{g_i} \} \\
 &~\mathrm{such ~that}~ \overrightarrow{g^{\prime}} : e^{\prime} \mapsto 0 ~\mathrm{and}~ \overleftarrow{g^{\prime}}: 0 \mapsto e^{\prime}.
 \end{aligned}
\end{equation}
\noindent
The unique value $e$ is the ``pointer value'' for the elementary QRF, represented as an $m$-bit string.  A nonelementary QRF that acts on $m$-bit strings is simply an ordered sequence of $m$ elementary QRFs acting on $m$-bit strings, subject to the uniqueness restriction that:

\begin{equation} \label{QRF4}
\forall j, ~1 \leq ~j \leq ~m,  ~ \exists ! ~e_j ~\mathrm{satisfying ~Eq. ~\eqref{QRF3a} ~and ~\eqref{QRF3b}}.
\end{equation}
\noindent
We can, therefore, define:

\begin{definition} \label{QRF-def}
A (deterministic) {\em quantum reference frame (QRF)} is a finite physical system $\mathbf{X}$ that implements some number $n \leq m$ of finite sets $\{ \overrightarrow{g_i} \}_j$ and $\{ \overleftarrow{g_i} \}_j$ of functions satisfying Eq. \eqref{QRF2} - \eqref{QRF4}.
\end{definition}
\noindent
A QRF is {\em elementary} if $n = 1$, i.e. the QRF measures and prepares one single value $e$.  A QRF is {\em fine-grained} if $n = m$; otherwise it is {\em coarse-grained}.  Elementary QRFs can be viewed as maximally coarse-grained.  We will refer, independently of $n$, to the number $m$ of bits in the input (or output) bit strings of a QRF $\mathbf{X}$ as its ``dimension'' (or ``resolution'') and write $m = \dim(\mathbf{X})$ for reasons that will become clear.  An $m$-dimensional QRF can report, given Eq. \eqref{QRF2} - \eqref{QRF4}, at most $m$ pointer values; this limit is achieved if the QRF is fine-grained.  We will in what follows assume all QRFs are either elementary or fine-grained, and that they are deterministic.  Relaxing Definition \ref{QRF-def} to nondeterministic QRFs requires mapping each bit string in $\{0, 1 \}^m$ to and from a probability value; details of this relaxation are provided in Appendix 1.
\footnote{As reported in \cite{fg:20,fgl-neurons} an example of a biological QRF is the Chey-Y system employed by bacterial chemotaxis as exhibited in e.g. {\em E Coli} where concentrations of phosphorylated Che-Y can influence ``swimming'' or ``tumbling'' analogous to a spin-up/spin-down mechanism \cite{micali:16}.
Similarly, the QRF formulism may be applicable to e.g. the hypothesis of the Posner molecule $\rm{Ca}_9(\rm{PO}_4)_9$ as an agent of enzyme catalyzed qubit computation involving nuclear spin entanglement \cite{fisher:15}. An explicit representation of a QRF in relationship to Hawking radiation is presented in \cite{fgm:21}.}

\subsection{CCCDs as specifications of computations} \label{CCCD-comp}

Having defined QRFs as computations, we now introduce CCCDs of Barwise-Seligman classifiers as generic specifications of computations.  A Barwise-Seligman \cite{barwise:97} classifier is a relation between tokens and types in some language:

\begin{definition}\label{class-def-1}
A {\em classifier} $\mathcal{A}$ is a triple $\langle Tok(\mathcal{A}), Typ(\mathcal{A}), \models_{\mathcal{A}} \rangle$ where $Tok(\mathcal{A})$ is a set of ``tokens'', $Typ(\mathcal{A})$ is a set of ``types'', and $\models_{\mathcal{A}}$ is a ``classification'' relating tokens to types.
\end{definition}
\noindent
The classification $\models_{\mathcal{A}}$ can, in general, be valued in any set $K$ without assumed structure; for simplicity, we will consider only binary classifications.  Morphisms between classifiers are called ``infomorphisms'' and are defined as:

\begin{definition}\label{class-def-2}
Given two classifiers $\mathcal{A} = \langle Tok(\mathcal{A}), Typ(\mathcal{A}), \models_{\mathcal{A}} \rangle$ and $\mathcal{B} = \langle Tok(\mathcal{B}), Typ(\mathcal{B}), \models_{\mathcal{B}} \rangle$, an {\em infomorphism} $f: \mathcal{A} \rightarrow \mathcal{B}$ is a pair of maps $\overrightarrow{f}: Tok(\mathcal{B}) \rightarrow Tok(\mathcal{A})$ and $\overleftarrow{f}: Typ(\mathcal{A}) \rightarrow Typ(\mathcal{B})$ such that $\forall b \in Tok(\mathcal{B})$ and $\forall a \in Typ(\mathcal{A})$, $\overrightarrow{f}(b) \models_{\mathcal{A}} a$ if and only if $b \models_{\mathcal{B}} \overleftarrow{f}(a)$.
\end{definition}
\noindent
This last definition can be represented schematically as the requirement that the following diagram commutes:

\begin{equation}\label{info-diagram-1}
\begin{gathered}
\xymatrix@!C=3pc{\rm{Typ}(\mathcal{A}) \ar[r]^{\overleftarrow{f}}   & \rm{Typ}(\mathcal{B}) \ar@{-}[d]^{\models_{\mathcal{B}}} \\
\rm{Tok}(\mathcal{A}) \ar@{-}[u]^{\models_{\mathcal{A}}}  & \rm{Tok}(\mathcal{B}) \ar[l]_{\overrightarrow{f}}}
\end{gathered}
\end{equation}
\noindent
An infomorphism $f: \mathcal{A} \rightarrow \mathcal{B}$ is, effectively, a map relating the semantic constraints imposed by the classification $\models_{\mathcal{A}}$ to those imposed by $\models_{\mathcal{B}}$.

Taking classifiers to be objects and infomorphisms to be arrows, we can define a category $\mathbf{Chan}$ (for ``Channel Theory'' \cite{barwise:97}) that is clearly isomorphic to the category $\mathbf{Chu}$ of Chu spaces \cite{barr:79, pratt:99a, pratt:99b} --- see \cite{fg:19a} for examples and discussion.  We will for simplicity restrict ourselves to binary classifiers as representations of yes/no questions and denote this restricted category by $\mathbf{Chan}$ in contexts without ambiguity.  It is straightforward to extend the theory to accommodate probability distributions (again see \cite{fg:19a}) --- details are provided in Appendix~1.

Consider now a collection $\mathcal{A}_i$ of $m$ binary classifiers, which we can take to operate, as an ordered array, on elements of $\{0, 1 \}^m$.  We can represent any Boolean operation on the outputs of these classifiers as a binary ``core'' classifier $\mathbf{C^\prime}$ that is the category-theoretic colimit of the underlying classifiers $\mathcal{A}_i$; the infomorphisms from the $\mathcal{A}_i$ to $\mathbf{C^\prime}$ then form a finite, commuting {\em cocone diagram} (CCD):

\begin{equation}\label{ccd-1}
\begin{gathered}
\xymatrix@C=4pc{&\mathbf{C^\prime} &  \\
\mathcal{A}_1 \ar[ur]^{f_1} \ar[r]_{g_{12}} & \mathcal{A}_2 \ar[u]_{f_2} \ar[r]_{g_{23}} & \ldots ~\mathcal{A}_k \ar[ul]_{f_k}
}
\end{gathered}
\end{equation}
\noindent
The dual of this construction is a commuting finite \emph{cone diagram} (CD) of infomorphisms on the same classifiers, where all arrows are reversed; here the ``core'' $\mathbf{D^\prime}$ is the category-theoretic limit of all possible downward-going structure-preserving maps to the classifiers $\mathcal{A}_i$.  Combining Diagram \eqref{ccd-1} with its CD dual yields a finite, commuting {\em cone-cocone diagram} (CCCD) on the single finite set of classifiers $\mathcal{A}_i$ linked by infomorphisms as depicted below:

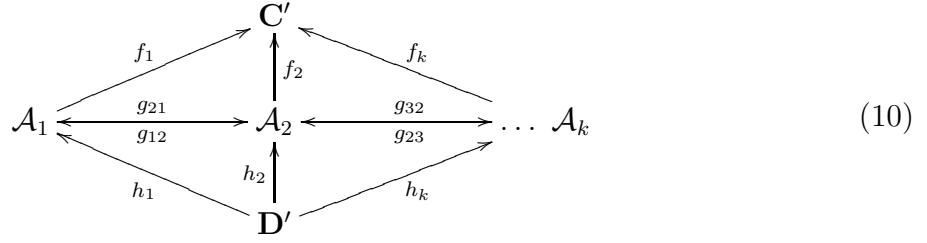
\begin{equation}\label{cccd-1}
\begin{gathered}
\xymatrix@C=6pc{&\mathbf{C^\prime} &  \\
\mathcal{A}_1 \ar[ur]^{f_1} \ar[r]_{g_{12}}^{g_{21}} & \ar[l] \mathcal{A}_2 \ar[u]_{f_2} \ar[r]_{g_{23}}^{g_{32}} & \ar[l] \ldots ~\mathcal{A}_k \ar[ul]_{f_k} \\
&\mathbf{D^\prime} \ar[ul]^{h_1} \ar[u]^{h_2} \ar[ur]_{h_k}&
}
\end{gathered}
\end{equation}
\noindent
If the cores $\mathbf{C^\prime} = \mathbf{D^\prime}$, we can also represent the CCCD as:

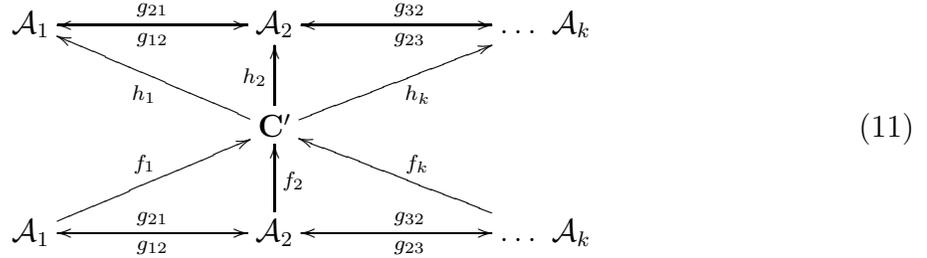
\begin{equation}\label{cccd-2}
\begin{gathered}
\xymatrix@C=6pc{\mathcal{A}_1 \ar[r]_{g_{12}}^{g_{21}} & \ar[l] \mathcal{A}_2 \ar[r]_{g_{23}}^{g_{32}} & \ar[l] \ldots ~\mathcal{A}_k \\
&\mathbf{C^\prime} \ar[ul]^{h_1} \ar[u]^{h_2} \ar[ur]_{h_k}& \\
\mathcal{A}_1 \ar[ur]^{f_1} \ar[r]_{g_{12}}^{g_{21}} & \ar[l] \mathcal{A}_2 \ar[u]_{f_2} \ar[r]_{g_{23}}^{g_{32}} & \ar[l] \ldots ~\mathcal{A}_k \ar[ul]_{f_k}
}
\end{gathered}
\end{equation}
\noindent
This diagram is naturally interpreted as an automorphism $\{0, 1 \}^m \rightarrow \{0, 1 \}^m$ implemented by passage through the constraint network having $\mathbf{C^\prime}$ as its apex.  When Boolean functions are extended to support probabilities as in Appendix~1, such systems become VAEs.  They can implement arbitrary Bayesian networks as discussed in \cite{fg:22, ffgl:22}.

\subsection{Specifying QRFs by CCCDs} \label{QRF-by-CCCD}

Comparing Diagram \eqref{cccd-2} with Eq. \eqref{QRF2}, it is clear that CCCDs and QRFs ``do the same thing.''  Indeed we can now state:

\begin{theorem} \label{thm1}
Every QRF can be specified by a CCCD.
\end{theorem}
\noindent
We approach this theorem by first proving it for the special case of elementary QRFs, and then showing how to construct arbitrary QRFs.  Hence we first state:

\begin{lemma} \label{lem1}
Every elementary QRF can be specified by a CCCD.
\end{lemma}

\begin{proof}
Let $\mathbf{X}$ be an elementary QRF with dimension $m$ and let $e \in \{0, 1 \}^m$ be the $m$-bit string encoding  its unique pointer value.  Consider now a CCCD in the form of Diagram \eqref{cccd-2}, with $m$ ``base level'' classifiers $\mathcal{A}_i$ mapping via $m$ infomorphisms $f_i$ to a core $\mathbf{C^\prime}$, and $m$ further infomorphisms $h_i$ mapping $\mathbf{C^\prime}$ back to the $\mathcal{A}_i$.  The only question at issue is whether Eq. \eqref{QRF3a} and \eqref{QRF3b} hold, i.e. whether the core $\mathbf{C^\prime}$ picks out a unique element $e \in \{0, 1 \}^m$.  To show this, note that the limit/colimit is unique if it exists.  Hence it picks out a unique $e$ provided it implements a conjunction of orthogonal constraints; i.e. some one of the $\mathcal{A}_i$ reports an outcome `1' and all others report `0'.  Each of the ``base level'' classifiers $\mathcal{A}_i$, however, acts on just one bit position in any input string $x \in \{0, 1 \}^m$; the constraints implemented by the $\mathcal{A}_i$ are, therefore, orthogonal, and Eq. \eqref{QRF3a} and \eqref{QRF3b} guarantee that one of the $\mathcal{A}_i$ reports an outcome `1' and all others report `0'.
\end{proof}
\noindent
It is natural to define the ``dimension'' of a CCCD representing an elementary QRF to be the number $m$ of base-level classifiers, i.e. as identical to the dimension of the QRF that it represents.

We now consider the construction of arbitrary $m$ dimensional QRFs from elementary QRFs specified as CCCDs of the form of Diagram \eqref{cccd-2}.  To simplify the notation, we will represent graphically only the CCD component of the CCCD as in Diagram  \eqref{ccd-1}; the CD component is, however, assumed in all cases to be present even though not depicted explicitly.  The duals of all depicted arrows are similarly assumed to be present though not depicted explicitly. Suppose a collection of $m$ distinct $m$-dimensional elementary QRFs $\mathbf{X}_j$, each specified by a CCCD having a core $\mathbf{C^\prime}_j$ that picks out a unique element $e_j \in \{0, 1 \}^m$ as illustrated below:

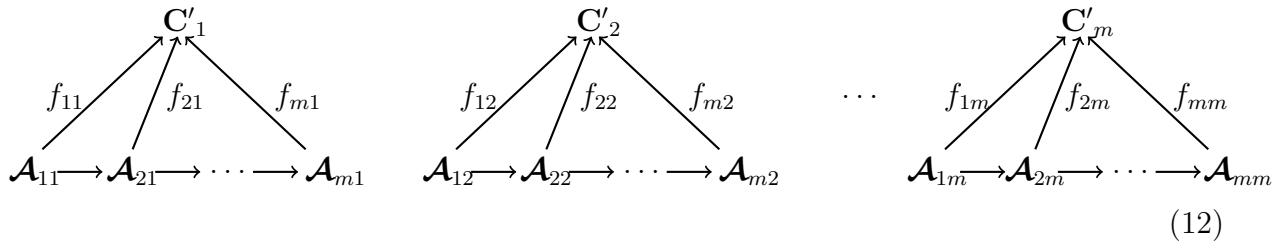
\begin{equation} \label{assembly-1}
\begin{gathered}
\begin{tikzpicture}
\node at (0,0) {$\mathbf{C^\prime}_1$};
\node at (-2,-2) {$\A_{11}$};
\node at (-0.7,-2) {$\A_{21}$};
\node at (0.6,-2) {$\dots$};
\node at (2,-2) {$\A_{m1}$};
\draw [thick, ->] (-1.7,-2) -- (-1.1,-2);
\draw [thick, ->] (-0.4,-2) -- (0.2,-2);
\draw [thick, ->] (0.9,-2) -- (1.5,-2);
\node at (-1.6,-1) {$f_{11}$};
\node at (0,-1) {$f_{21}$};
\node at (1.5,-1) {$f_{m1}$};
\draw [thick, ->] (-1.9,-1.7) -- (-0.3,-0.2);
\draw [thick, ->] (-0.7,-1.7) -- (-0.1,-0.2);
\draw [thick, ->] (1.6,-1.7) -- (0,-0.2);
\node at (5.5,0) {$\mathbf{C^\prime}_2$};
\node at (3.5,-2) {$\A_{12}$};
\node at (4.8,-2) {$\A_{22}$};
\node at (6.1,-2) {$\dots$};
\node at (7.5,-2) {$\A_{m2}$};
\draw [thick, ->] (3.8,-2) -- (4.4,-2);
\draw [thick, ->] (5.1,-2) -- (5.7,-2);
\draw [thick, ->] (6.4,-2) -- (7,-2);
\node at (3.9,-1) {$f_{12}$};
\node at (5.5,-1) {$f_{22}$};
\node at (7,-1) {$f_{m2}$};
\draw [thick, ->] (3.6,-1.7) -- (5.2,-0.2);
\draw [thick, ->] (4.8,-1.7) -- (5.4,-0.2);
\draw [thick, ->] (7.1,-1.7) -- (5.5,-0.2);
\node at (9,-1) {$\dots$};
\node at (12,0) {$\mathbf{C^\prime}_m$};
\node at (10,-2) {$\A_{1m}$};
\node at (11.3,-2) {$\A_{2m}$};
\node at (12.6,-2) {$\dots$};
\node at (14,-2) {$\A_{mm}$};
\draw [thick, ->] (10.3,-2) -- (10.9,-2);
\draw [thick, ->] (11.6,-2) -- (12.2,-2);
\draw [thick, ->] (12.9,-2) -- (13.5,-2);
\node at (10.4,-1) {$f_{1m}$};
\node at (12,-1) {$f_{2m}$};
\node at (13.5,-1) {$f_{mm}$};
\draw [thick, ->] (10.1,-1.7) -- (11.7,-0.2);
\draw [thick, ->] (11.3,-1.7) -- (11.9,-0.2);
\draw [thick, ->] (13.6,-1.7) -- (12,-0.2);
\end{tikzpicture}
\end{gathered}
\end{equation}
\noindent
Intuitively, the pointer outcomes reported by any nonelementary QRF $\mathbf{X}$ represent the possible values, coarse-grained to the resolution of the QRF, of some physical observable; the tick marks on a meter stick, for example, represent possible values of length.  All measurements of any fixed observable commute; this is simply the fact that the operators $M^k_i$ in Eq. \eqref{ham} all mutually commute for any fixed choice of basis for the system $k$.  If, therefore, the elementary QRFs $\mathbf{X}_j$ all report pointer outcomes for the same physical observable, they must all mutually commute.  The corresponding CCCDs depicted in Diagram \eqref{assembly-1} must, therefore, all mutually commute, i.e. defining additional infomorphisms between any of the $\A_{ij}$ always yields a commutative diagram.  Under these conditions, which indeed correspond to quantum noncontextuality, the existence of a common core $\mathbf{C}$ and maps $\psi_j : \mathbf{C^\prime}_j \mapsto \mathbf{C}$ is guaranteed by \cite[Thm. 7.1]{fg:22}.  We therefore have the CCCD:

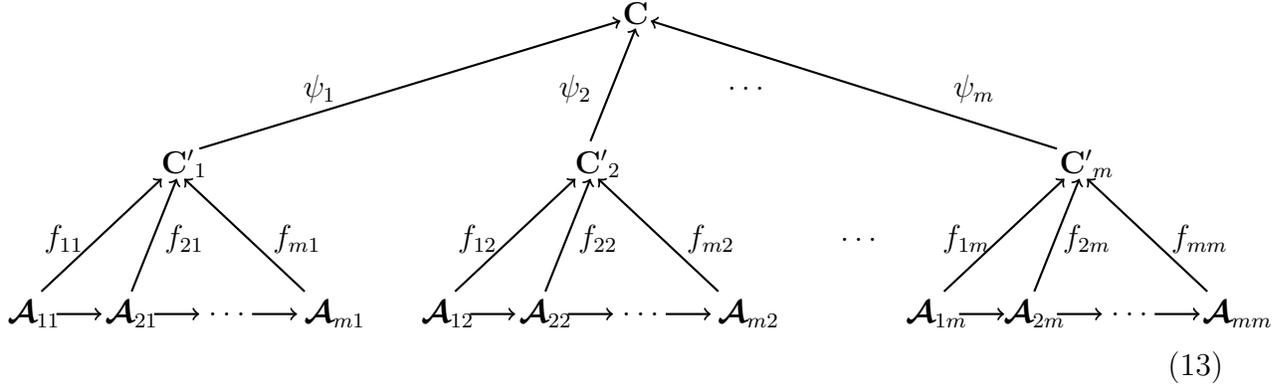
\begin{equation} \label{assembly-2}
\begin{gathered}
\begin{tikzpicture}
\node at (0,0) {$\mathbf{C^\prime}_1$};
\node at (-2,-2) {$\A_{11}$};
\node at (-0.7,-2) {$\A_{21}$};
\node at (0.6,-2) {$\dots$};
\node at (2,-2) {$\A_{m1}$};
\draw [thick, ->] (-1.7,-2) -- (-1.1,-2);
\draw [thick, ->] (-0.4,-2) -- (0.2,-2);
\draw [thick, ->] (0.9,-2) -- (1.5,-2);
\node at (-1.6,-1) {$f_{11}$};
\node at (0,-1) {$f_{21}$};
\node at (1.5,-1) {$f_{m1}$};
\draw [thick, ->] (-1.9,-1.7) -- (-0.3,-0.2);
\draw [thick, ->] (-0.7,-1.7) -- (-0.1,-0.2);
\draw [thick, ->] (1.6,-1.7) -- (0,-0.2);
\node at (5.5,0) {$\mathbf{C^\prime}_2$};
\node at (3.5,-2) {$\A_{12}$};
\node at (4.8,-2) {$\A_{22}$};
\node at (6.1,-2) {$\dots$};
\node at (7.5,-2) {$\A_{m2}$};
\draw [thick, ->] (3.8,-2) -- (4.4,-2);
\draw [thick, ->] (5.1,-2) -- (5.7,-2);
\draw [thick, ->] (6.4,-2) -- (7,-2);
\node at (3.9,-1) {$f_{12}$};
\node at (5.5,-1) {$f_{22}$};
\node at (7,-1) {$f_{m2}$};
\draw [thick, ->] (3.6,-1.7) -- (5.2,-0.2);
\draw [thick, ->] (4.8,-1.7) -- (5.4,-0.2);
\draw [thick, ->] (7.1,-1.7) -- (5.5,-0.2);
\node at (9,-1) {$\dots$};
\node at (12,0) {$\mathbf{C^\prime}_m$};
\node at (10,-2) {$\A_{1m}$};
\node at (11.3,-2) {$\A_{2m}$};
\node at (12.6,-2) {$\dots$};
\node at (14,-2) {$\A_{mm}$};
\draw [thick, ->] (10.3,-2) -- (10.9,-2);
\draw [thick, ->] (11.6,-2) -- (12.2,-2);
\draw [thick, ->] (12.9,-2) -- (13.5,-2);
\node at (10.4,-1) {$f_{1m}$};
\node at (12,-1) {$f_{2m}$};
\node at (13.5,-1) {$f_{mm}$};
\draw [thick, ->] (10.1,-1.7) -- (11.7,-0.2);
\draw [thick, ->] (11.3,-1.7) -- (11.9,-0.2);
\draw [thick, ->] (13.6,-1.7) -- (12,-0.2);
\node at (6,2) {$\mathbf{C}$};
\draw [thick, ->] (0.2,0.2) -- (5.8,1.9);
\draw [thick, ->] (5.4,0.3) -- (6,1.8);
\draw [thick, ->] (11.6,0.2) -- (6.2,1.9);
\node at (7.5,1) {$\dots$};
\node at (1.8,1) {$\psi_1$};
\node at (5.2,1) {$\psi_2$};
\node at (10.5,1) {$\psi_m$};
\end{tikzpicture}
\end{gathered}
\end{equation}
\noindent
We can now show that the core $\mathbf{C}$ can be chosen so that Diagram \eqref{assembly-2} represents a fine-grained, $m$-dimensional QRF $\mathbf{X}$ comprising the $\mathbf{X}_j$.

\begin{proof}[Proof {\rm (Theorem \ref{thm1})}]
By Lemma \ref{lem1} we have that each of the component CCCDs shown in Diagram \eqref{assembly-1} specifies an elementary QRF $\mathbf{X}_j$ reporting a pointer value $e_j \in \{0, 1 \}^m$; both the $\mathbf{X}_j$ and the reported $e_j$ are, moreover, distinct by construction.   The maps $\psi_j$ all commute by \cite[Thm. 7.1]{fg:22}; hence Diagram \eqref{assembly-2} is a CCCD.  When presented with any bit string in $\{0, 1 \}^m$, at most one of the cores $\mathbf{C^\prime}_j$ will encode an outcome ``1''; this will occur if but only if the presented bit string is identical to $e_j$.  Hence by specifying that the core $\mathbf{C}$ computes the disjunct of the values encoded by the $\mathbf{C^\prime}_j$, Diagram \eqref{assembly-2} specifies the required fine-grained, $m$-dimensional QRF $\mathbf{X}$.
\end{proof}
Diagram \eqref{assembly-2} can, clearly, be redrawn in the form of Diagram \eqref{cccd-2}, again respecting the convention of leaving the CD components implicit, by composing each infomorphism $f_{ij}$ with the appropriate $\psi_j$ to construct ``direct'' maps from each of the $\A_{ij}$ to the core $\mathbf{C}$.  However, preserving the interpretation of Diagram \eqref{assembly-2} as representing an $m$-dimensional QRF under this redrawing requires redefining $\mathbf{C}$ as a classifier; whereas $\mathbf{C}$ computes a disjunction of its inputs in Diagram \eqref{assembly-2}, it must compute a disjunction, over the index $j$, of conjuncts of the outcomes of the $m$ composed maps $\psi_j f_{ij}$ in the redrawn diagram.  This preserves the function (i.e. conjunction of their $m$ inputs) of the elementary QRF cores $\mathbf{C^\prime}_i$ that are replaced by families of composite infomorphisms in the redrawn diagram.  Hence it preserves the sense in which the dimension $m$ of each elementary QRF, and of the CCCD that represents it, can be considered the ``dimension'' of the composite CCCD; computing a disjunct of $m$ bit strings yields no more than an $m$ bit string.  To generalize, we can add any number of intermediate ``core'' classifiers representing limits/colimits over particular subsets of the $\A_{ij}$ to Diagram \eqref{assembly-2} without violating commutativity constraints, but we must assure that the intermediate cores and $\mathbf{C}$ are (re)defined as classifiers so as to preserve the overall logical operation of the CCCD.  Hence we can say, in general, that intermediate ``layers'' of classifiers can be added to or subtracted from a CCCD provided that the function that it computes and hence the QRF that it represents remains constant.  This allows us to employ the compact notation:

\begin{equation} \label{CCCD-2}
\begin{gathered}
\begin{tikzpicture}
\draw [thick] (0,0) -- (2,1) -- (2,-1) -- (0,0);
\node at (1.3,0) {$\mathbf{X}$};
\end{tikzpicture}
\end{gathered}
\end{equation}
\noindent
for an $m$-dimensional CCCD specifying an $m$-dimensional QRF $\mathbf{X}$.  Here we regard the left-facing apex of the triangle as depicting the core, and the right-facing base of the triangle as depicting the $m$ arrays (or in a coarse-grained QRF, $n < m$ arrays) of $m$ base-level classifiers $\A_{ij}$.  In this notation, a QRF $\mathbf{S} = \mathbf{S_1 S_2}$ becomes:

\begin{equation} \label{CCCD-3}
\begin{gathered}
\begin{tikzpicture}
\node at (0,0) {$\mathbf{S}$};
\draw [thick] (0.2,0) -- (1,1) -- (2,1.5) -- (2,0.5) -- (1,1);
\node at (1.7,1) {$\mathbf{S_1}$};
\draw [thick] (0.2,0) -- (1,-1) -- (2,-0.5) -- (2,-1.5) -- (1,-1);
\node at (1.7,-1) {$\mathbf{S_2}$};
\end{tikzpicture}
\end{gathered}
\end{equation}
\noindent
Note that this diagram implies vertical arrows in both directions between $\mathbf{S_1}$ and $\mathbf{S_2}$ to assure commutativity of the CCCD.

\section{Sequential measurements} \label{CCCD-cat}

To discuss sequential measurements of some sector $S$ of $\mathscr{B}$, we must understand how QRFs can be deployed, alone or in combination, at sequential times, subject to the constraint that all QRFs co-deployed at any particular time all mutually commute.  Representing QRFs as CCCDs, such time evolutions become morphisms of CCCDs, which clearly must compose, and must include a unique Identity morphism for each CCCD. This motivates the following:

\begin{definition} \label{def-cat-CCCD}
The category $\mathbf{CCCD}$ is defined to consist of, as objects, all diagrams of CCCDs, considered as objects of the category $\mathbf{Digraph}$ of (finite) directed graphs, i.e. all diagrams of the form of Diagram \eqref{cccd-2} comprising finite sets of classifiers and their associated limits/colimits, and as arrows, morphisms of $\mathbf{Digraph}$ connecting such objects.
\end{definition}

Morphisms so defined clearly compose, and there is for every CCCD a unique Identity morphism.
Heuristically, this definition ties implicitly to a functor\footnote{Alternatively, the CCD in \eqref{ccd-1} comprises the cone category $\mathbf{Cone}(\mathbf{C'})$ of arrows into $\mathbf{C'}$ --- see e.g. \cite[p.102]{awodey:10}.  The category pertaining to the CD can be defined as the opposite category $\mathbf{Cone}(\mathbf{C'})^{(\rm{opp})}$, i.e. reversing the arrows to make a cone category of arrows from $\mathbf{D'}$ in \eqref{cccd-1}.  One then forms the product category $\mathbf{CCCD} = \mathbf{Cone}(\mathbf{C'}) \times \mathbf{Cone}(\mathbf{C'})^{(\rm{opp})}$ ---  i.e. ordered pairs of objects and arrows in each factor.} $\mathbf{Digraph} \lra \mathbf{CCCD}$.

The above Definition \ref{def-cat-CCCD} places no restriction on the semantics of CCCDs or the functions that they compute.  The CCCDs of interest, however, are those that represent QRFs, and hence have the form of Diagram \eqref{assembly-2}.  As noted above, any such CCCD has a well-defined dimension, which is equal to the dimension of the QRF it represents.  Morphisms of such CCCDs, as defined, can clearly produce arbitrarily large structures computing arbitrarily complex Boolean functions, that represent arbitrarily large, multi-variate QRFs.  For the present purposes, however, we will only be concerned with CCCD morphisms that implement sequential measurements on a single system.  It is sufficient to consider the two types of sequential measurements mentioned above, fission/fusion measurements and switches between noncommuting pointer QRFs, which can also be considered context switches.  The first type is exemplified by:

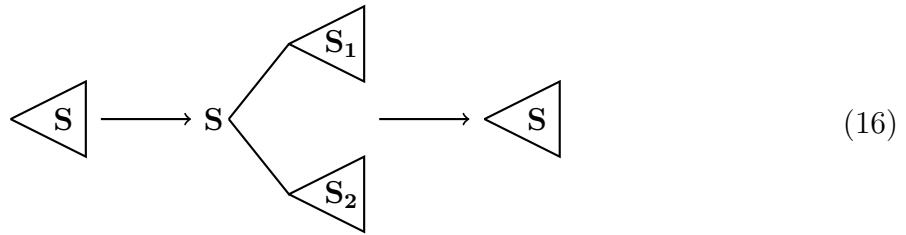
\begin{equation} \label{flow-1}
\begin{gathered}
\begin{tikzpicture}
\node at (0,0) {$\mathbf{S}$};
\draw [thick] (0.2,0) -- (1,1) -- (2,1.5) -- (2,0.5) -- (1,1);
\node at (1.7,1) {$\mathbf{S_1}$};
\draw [thick] (0.2,0) -- (1,-1) -- (2,-0.5) -- (2,-1.5) -- (1,-1);
\node at (1.7,-1) {$\mathbf{S_2}$};
\draw [thick] (-2.7,0) -- (-1.7,0.5) -- (-1.7,-0.5) -- (-2.7,0);
\node at (-2,0) {$\mathbf{S}$};
\draw [thick, ->] (-1.5,0) -- (-0.3,0);
\draw [thick, ->] (2.2,0) -- (3.4,0);
\draw [thick] (3.6,0) -- (4.6,0.5) -- (4.6,-0.5) -- (3.6,0);
\node at (4.3,0) {$\mathbf{S}$};
\end{tikzpicture}
\end{gathered}
\end{equation}
\noindent
For simplicity, we will take the QRF $\mathbf{S}$ to be elementary and to have dimension $m$.  In this case, Diagram \eqref{flow-1} just represents a relabelling of subsets of the $m$ base-level classifiers:

\begin{equation} \label{class-flow-1}
\underbrace{\mathcal{A}_1, \mathcal{A}_2, \dots \mathcal{A}_m}_{\mathbf{S}} \rightarrow \underbrace{\mathcal{A}_1, \dots \mathcal{A}_i,}_{\mathbf{S_1}} \underbrace{\mathcal{A}_{i+1}, \dots \mathcal{A}_m}_{\mathbf{S_2}} \rightarrow \underbrace{\mathcal{A}_1, \mathcal{A}_2, \dots \mathcal{A}_m}_{\mathbf{S}}
\end{equation}
\noindent
that allows defining limits/colimits $\mathbf{S_1}$ and $\mathbf{S_2}$ over the $\mathcal{A}_1, \dots \mathcal{A}_i$ and the $\mathcal{A}_{i+1}, \dots \mathcal{A}_m$, respectively.  Because the limit/colimit $\mathbf{S}$ exists over the $\mathcal{A}_1, \mathcal{A}_2, \dots \mathcal{A}_m$, it clearly exists over $\mathbf{S_1}$ and $\mathbf{S_2}$.

In the second type of example, the initial and final measurements are those ones of the reference component $R$ only; the pointer component is traced over \cite{fields:19}.  These reference measurements serve to confirm the identity of $S$ and hence assure that the subsequent measurements are indeed sequential measurements of one and the same system.  We can represent the sequence as:

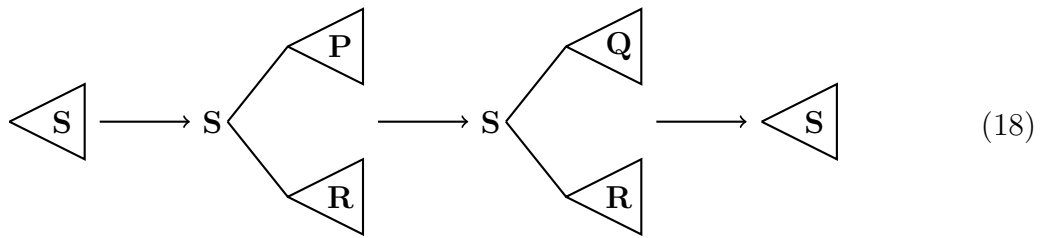
\begin{equation} \label{flow-2}
\begin{gathered}
\begin{tikzpicture}
\node at (0,0) {$\mathbf{S}$};
\draw [thick] (0.2,0) -- (1,1) -- (2,1.5) -- (2,0.5) -- (1,1);
\node at (1.7,1) {$\mathbf{P}$};
\draw [thick] (0.2,0) -- (1,-1) -- (2,-0.5) -- (2,-1.5) -- (1,-1);
\node at (1.7,-1) {$\mathbf{R}$};
\draw [thick] (-2.7,0) -- (-1.7,0.5) -- (-1.7,-0.5) -- (-2.7,0);
\node at (-2,0) {$\mathbf{S}$};
\draw [thick, ->] (-1.5,0) -- (-0.3,0);
\draw [thick, ->] (2.2,0) -- (3.4,0);
\node at (3.7,0) {$\mathbf{S}$};
\draw [thick] (3.9,0) -- (4.7,1) -- (5.7,1.5) -- (5.7,0.5) -- (4.7,1);
\node at (5.4,1) {$\mathbf{Q}$};
\draw [thick] (3.9,0) -- (4.7,-1) -- (5.7,-0.5) -- (5.7,-1.5) -- (4.7,-1);
\node at (5.4,-1) {$\mathbf{R}$};
\draw [thick, ->] (5.9,0) -- (7.1,0);
\draw [thick] (7.3,0) -- (8.3,0.5) -- (8.3,-0.5) -- (7.3,0);
\node at (8,0) {$\mathbf{S}$};
\end{tikzpicture}
\end{gathered}
\end{equation}
\noindent
Again taking $\mathbf{S}$ to be elementary with dimension $m$, and letting $k < m$, we can represent this in terms of classifier labels, leaving traced-over classifiers implicit, as:

\begin{equation} \label{class-flow-2}
\underbrace{\mathcal{A}_1, \mathcal{A}_2, \dots \mathcal{A}_k}_{\mathbf{R}} \rightarrow \underbrace{\mathcal{A}_1, \dots \mathcal{A}_k,}_{\mathbf{R}} \underbrace{\mathcal{A}_{k+1}, \dots \mathcal{A}_m}_{\mathbf{P}} \rightarrow \underbrace{\mathcal{A}_1, \dots \mathcal{A}_k,}_{\mathbf{R}} \underbrace{\tilde{\mathcal{A}}_{k+1}, \dots \tilde{\mathcal{A}}_m}_{\mathbf{Q}} \rightarrow \underbrace{\mathcal{A}_1, \mathcal{A}_2, \dots \mathcal{A}_k}_{\mathbf{R}}
\end{equation}
\noindent
where the notation $\tilde{\mathcal{A}}_{l}$ indicates that $\mathcal{A}_{l}$ has been rewritten in a rotated measurement basis, e.g. $s_z \rightarrow s_x$ or $x \rightarrow p= m\, (\partial x/ \partial t)$.  As with Eq. \eqref{class-flow-1}, Eq. \eqref{class-flow-2} is a relabeling, but here the local basis rotation mapping $\mathbf{P} \rightarrow \mathbf{Q}$ has been added.  As both $\mathbf{P}$ and $\mathbf{Q}$ must commute with $\mathbf{R}$, the commutativity requirements for $\mathbf{S}$ are satisfied.

\section{Mapping CCCDs to cobordisms} \label{CCCD2cobord}

We are now in a position to map sequential measurements as represented by CCCDs to cobordisms of finite Hilbert spaces and hence to TQFTs.  To proceed, consider now an $m$-dimensional CCCD $\mathbf{X}$, and associate with it the set $\{ 0,1 \}^m$ of bit strings on which $\mathbf{X}$ acts as a computation.
Diagram \ref{prop} lets us associate with $\{ 0,1 \}^m$ a $m$-dimensional sector $X$ in some space (with discrete topology) $\mathscr{B}$ of dimension $N \ll m$ and also with an $m$-dimensional Hilbert space $\mathcal{H}_X$.
Hence we can represent Diagrams \eqref{flow-1} and \eqref{flow-2} as, respectively, maps
\begin{equation}\label{sector-map-1}
S \rightarrow S_1 \times S_2 \rightarrow S ~\text{and}~ S \rightarrow P \times R \rightarrow Q \times R \rightarrow S
\end{equation}
on sectors of some appropriate $\mathscr{B}$, respecting the structure of respective objects and arrows in the category $\mathbf{CCCD}$.
We next correspond the maps of sectors in \eqref{sector-map-1} to proper, linear maps
\begin{equation}\label{hilbert-map-1}
\mathcal{H}_S \rightarrow \mathcal{H}_{S_1} \bigotimes \mathcal{H}_{S_2} \rightarrow \mathcal{H}_S~ \text{and}~ \mathcal{H}_S \rightarrow \mathcal{H}_P \bigotimes \mathcal{H}_R \rightarrow \mathcal{H}_Q \bigotimes \mathcal{H}_R \rightarrow \mathcal{H}_S
\end{equation}
on finite-dimensional Hilbert spaces.  It is natural to consider these latter maps as leading to a cobordism of Hilbert
spaces. To see this, we note:
\begin{lemma} \label{Fredholm}
The sequences of maps in \eqref{hilbert-map-1} are sequences of proper Fredholm maps of separable Hilbert manifolds.
\end{lemma}
\begin{proof}
Formally, such finite dimensional (hence separable) Hilbert spaces are naturally separable Hilbert manifolds on which a smooth structure is always definable
\cite{lang:72,lang:95}. The finite dimensionality of the Hilbert spaces in question implies that every linear map, and compositions of such, are Fredholm maps --- i.e
having finite dimensional kernel and cokernel --- and likewise for every smooth map of finite dimensional Hilbert manifolds \cite{zeidler:95} --- this has been also noted in e.g. \cite[Props. 1.6, 1.7]{baker:00}.
\end{proof}
The cobordism in our case can then be seen as a special case of the general theory of cobordisms of separable Hilbert manifolds following the details in \cite{baker:00}.

Lemma \ref{Fredholm} demonstrates a compatibility between sequential measurements of sectors, unique in each case, and their associated aforementioned maps of Hilbert spaces respecting composition, which lets us construct diagrams as shown below that implicitly define sequences of maps $\mathfrak{F} = \{\mathfrak{F}(i)\}$ acting at measurement times $i$.  In the case of fission/fusion measurements, e.g. Eq. \eqref{flow-1}, we have diagrams of the form

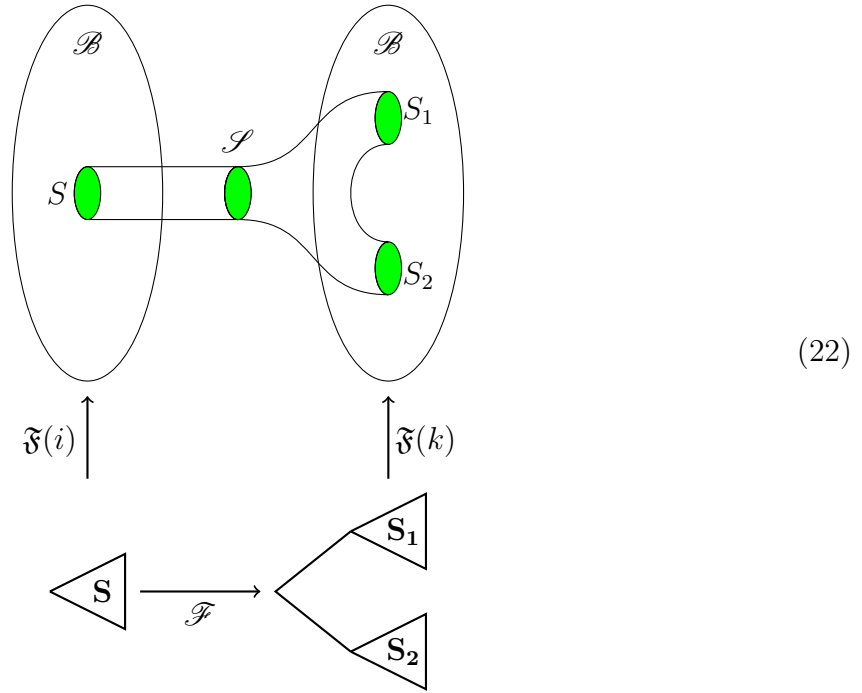
\begin{equation} \label{CCCD-to-Cob}
\begin{gathered}
\begin{tikzpicture}[every tqft/.append style={transform shape}]
\draw[rotate=90] (0,0) ellipse (2.5cm and 1 cm);
\node[above] at (0,1.7) {$\mathscr{B}$};
\node at (-0.4,0) {$S$};
\begin{scope}[tqft/every boundary component/.style={draw,fill=green,fill opacity=1}]
\begin{scope}[tqft/cobordism/.style={draw}]
\begin{scope}[rotate=90]
\pic[tqft/cylinder, name=a];
\pic[tqft/pair of pants, anchor=incoming boundary 1, name=b, at=(a-outgoing boundary 1)];
\end{scope}
\end{scope}
\end{scope}
\draw[rotate=90] (0,-4) ellipse (2.5cm and 1 cm);
\node[above] at (4,1.7) {$\mathscr{B}$};
\node at (2,0.7) {$\mathscr{S}$};
\node at (4.4,1.1) {$S_1$};
\node at (4.4,-1.1) {$S_2$};
\draw [thick, <-] (0,-2.7) -- (0,-3.8);
\draw [thick, <-] (4,-2.7) -- (4,-3.8);
\draw [thick] (-0.5,-5.3) -- (0.5,-4.8) -- (0.5,-5.8) -- (-0.5,-5.3);
\node at (0.2,-5.3) {$\mathbf{S}$};
\draw [thick] (3.5,-4.5) -- (4.5,-4) -- (4.5,-5) -- (3.5,-4.5);
\node at (4.2,-4.5) {$\mathbf{S_1}$};
\draw [thick] (3.5,-6.1) -- (4.5,-5.6) -- (4.5,-6.6) -- (3.5,-6.1);
\node at (4.2,-6.1) {$\mathbf{S_2}$};
\draw [thick] (3.5,-4.5) -- (2.5,-5.3) -- (3.5,-6.1);
\draw [thick, ->] (0.7,-5.3) -- (2.3,-5.3);
\node at (-0.5,-3.3) {$\mathfrak{F}(i)$};
\node at (4.5,-3.3) {$\mathfrak{F}(k)$};
\node at (1.5,-5.6) {$\mathscr{F}$};
\end{tikzpicture}
\end{gathered}
\end{equation}

while in the case of basis rotations, e.g. Eq. \eqref{flow-2}, we have diagrams of the form below, which make explicit the possibility of entanglement between reference and pointer components during periods of non-observation \cite{fields:19}, corresponding in traditional terms to ``the system of interest becoming entangled with the apparatus'':

\begin{equation} \label{CCCD-to-Cob-2}
\begin{gathered}
\begin{tikzpicture}[every tqft/.append style={transform shape}]
\draw[rotate=90] (0,0) ellipse (2.5cm and 1 cm);
\node[above] at (0,1.7) {$\mathscr{B}$};
\node at (-0.5,1.1) {$P$};
\node at (-0.5,-1.1) {$R$};
\begin{scope}[tqft/every boundary component/.style={draw,fill=green,fill opacity=1}]
\begin{scope}[tqft/cobordism/.style={draw}]
\begin{scope}[rotate=90]
\pic[tqft/reverse pair of pants, at={(-1,0)}, name=a];
\pic[tqft/pair of pants, anchor=incoming boundary 1, name=b, at=(a-outgoing boundary 1)];
\end{scope}
\end{scope}
\end{scope}
\draw[rotate=90] (0,-4) ellipse (2.5cm and 1 cm);
\node[above] at (4,1.7) {$\mathscr{B}$};
\node at (2,0.7) {$\mathscr{S}$};
\node at (4.4,1.1) {$Q$};
\node at (4.4,-1.1) {$R$};
\draw [thick, <-] (0,-2.7) -- (0,-3.8);
\draw [thick, <-] (4,-2.7) -- (4,-3.8);
\draw [thick] (-0.5,-4.5) -- (0.5,-4) -- (0.5,-5) -- (-0.5,-4.5);
\node at (0.2,-4.5) {$\mathbf{P}$};
\draw [thick] (-0.5,-6.1) -- (0.5,-5.6) -- (0.5,-6.6) -- (-0.5,-6.1);
\node at (0.2,-6.1) {$\mathbf{R}$};
\draw [thick] (-0.5,-4.5) -- (-1.6,-5.3) -- (-0.5,-6.1);
\node at (-1.8,-5.3) {$\mathbf{S}$};
\draw [thick] (3.5,-4.5) -- (4.5,-4) -- (4.5,-5) -- (3.5,-4.5);
\node at (4.2,-4.5) {$\mathbf{Q}$};
\draw [thick] (3.5,-6.1) -- (4.5,-5.6) -- (4.5,-6.6) -- (3.5,-6.1);
\node at (4.2,-6.1) {$\mathbf{R}$};
\draw [thick] (3.5,-4.5) -- (2.5,-5.3) -- (3.5,-6.1);
\draw [thick, ->] (0.7,-5.3) -- (2.3,-5.3);
\node at (-0.5,-3.3) {$\mathfrak{F}(i)$};
\node at (4.5,-3.3) {$\mathfrak{F}(k)$};
\node at (1.5,-5.6) {$\mathscr{F}$};
\end{tikzpicture}
\end{gathered}
\end{equation}

Diagrams \eqref{CCCD-to-Cob} and \eqref{CCCD-to-Cob-2} illustrate the compatibility between objects (representing QRFs) and arrows (representing sequential measurements) in $\mathbf{CCCD}$ and the resulting cobordisms on sectors, both respecting composition and identities.  Thus they demonstrate an induced functor $\mathfrak{F}: \mathbf{CCCD} \rightarrow \mathbf{Cob}$, the category of  finite cobordisms, in particular picking out cobordisms on finite Hilbert spaces in relationship to \eqref{hilbert-map-1}.

\begin{theorem} \label{thm2}
For any morphism $\mathscr{F}$ of CCCDs in $\mathbf{CCCD}$, there is a cobordism $\mathscr{S}$ such that a diagram of the form of Diagram \eqref{CCCD-to-Cob} or \eqref{CCCD-to-Cob-2} commutes.
\end{theorem}

As any morphism $\mathscr{F}$ of CCCDs can be constructed by iterating the morphisms shown in Diagrams \eqref{flow-1} and/or \eqref{flow-2}, it is sufficient to show that these morphisms map to cobordisms as required by Theorem \ref{thm2}.  Hence it is sufficient to chase around Diagrams \eqref{CCCD-to-Cob} and \eqref{CCCD-to-Cob-2} in both horizontal directions to show that commutativity holds.

\begin{proof}
Let $m$ be the larger of the CCCD dimensions, choose an $N \gg n$, let $\mathcal{H}_A$ be an $N$-dimensional Hilbert space and $\mathscr{B}$ be an $N$-dimensional space with discrete topology as in Diagram \eqref{prop}.  For each CCCD $\mathbf{X}$, let $\mathcal{H}_{\mathbf{X}}$ be the associated Hilbert space.  We can clearly embed any of these $\mathcal{H}_{\mathbf{X}}$ in $\mathcal{H}_A$. For any embedded $\mathcal{H}_{\mathbf{X}}$ and $\mathcal{H}_{\mathbf{Y}}$, the map $\mathcal{H}_{\mathbf{X}} \rightarrow \mathcal{H}_{\mathbf{Y}}$ comprises components of, and can be extended to, the Identity $\mathcal{H}_A \rightarrow \mathcal{H}_A$.  Hence this map itself specifies the required cobordism.
\end{proof}

Note that a structure such as

\begin{equation} \label{not-CCCD}
\begin{gathered}
\begin{tikzpicture}
\node at (0,0) {$\mathbf{!!}$};
\draw [thick] (0.2,0) -- (1,1) -- (2,1.5) -- (2,0.5) -- (1,1);
\node at (1.7,1) {$\mathbf{P}$};
\draw [thick] (0.2,0) -- (1,-1) -- (2,-0.5) -- (2,-1.5) -- (1,-1);
\node at (1.7,-1) {$\mathbf{Q}$};
\end{tikzpicture}
\end{gathered}
\end{equation}
\noindent
with $[\mathbf{P}, \mathbf{Q}] \neq 0$ does not commute, and therefore is not a CCCD by \cite[Thm. 7.1]{fg:22}. Hence it does not represent an operation on the space $\mathcal{H}$.

\section{Mapping sequential measurements to TQFTs}

The Identity map $\mathscr{B} \rightarrow \mathscr{B}$ is implemented by a unitary evolution operator $\mathcal{P}_U$ as in Diagram \eqref{prop}.  Hence we have a TQFT:

\begin{equation} \label{full-TQFT}
\begin{gathered}
\begin{tikzpicture}
\draw[rotate=90] (0,0) ellipse (2.5cm and 1 cm);
\node[above] at (0,1.7) {$\mathscr{B}$};
\draw[rotate=90] (0,-4) ellipse (2.5cm and 1 cm);
\node[above] at (4,1.7) {$\mathscr{B}$};
\node at (2,2.8) {$\mathscr{S}$};
\draw [thick] (0,2.5) -- (4,2.5);
\draw [thick] (0,-2.5) -- (4,-2.5);
\draw [thick, ->] (0,-2.7) -- (0,-3.7);
\draw [thick, ->] (4,-2.7) -- (4,-3.7);
\node at (-0.1,-4) {$\mathcal{H}_A|_i$};
\node at (3.9,-4) {$\mathcal{H}_A|_k$};
\draw [thick, ->] (0.5,-4) -- (3.3,-4);
\node at (1.9,-3.7) {$\mathcal{P}_U$};
\end{tikzpicture}
\end{gathered}
\end{equation}
\noindent
that maps the cobordism $\mathscr{S}$ of $\mathscr{B}$ to $A$'s apparent Hilbert space $\mathcal{H}_A$ describing bit-string encodings on $\mathscr{B}$, as measured at ``times'' $i$ and $k$.  Hence the proof of Theorem \ref{thm2} itself defines a TQFT.  Similarly, the diagrams:

\begin{equation} \label{S-TQFT}
\begin{gathered}
\begin{tikzpicture}[every tqft/.append style={transform shape}]
\draw[rotate=90] (0,0) ellipse (2.5cm and 1 cm);
\node[above] at (0,1.7) {$\mathscr{B}$};
\node at (-0.4,0) {$S$};
\begin{scope}[tqft/every boundary component/.style={draw,fill=green,fill opacity=1}]
\begin{scope}[tqft/cobordism/.style={draw}]
\begin{scope}[rotate=90]
\pic[tqft/cylinder, name=a];
\pic[tqft/pair of pants, anchor=incoming boundary 1, name=b, at=(a-outgoing boundary 1)];
\end{scope}
\end{scope}
\end{scope}
\draw[rotate=90] (0,-4) ellipse (2.5cm and 1 cm);
\node[above] at (4,1.7) {$\mathscr{B}$};
\node at (2,0.7) {$\mathscr{S}$};
\node at (4.4,1.1) {$S_1$};
\node at (4.4,-1.1) {$S_2$};
\draw [thick, ->] (0,-2.7) -- (0,-3.8);
\draw [thick, ->] (4,-2.7) -- (4,-3.8);
\node at (0,-4.2) {$\mathcal{H}_S|_i$};
\node at (4.2,-4.2) {$\mathcal{H}_{S_1}|_k \bigotimes \mathcal{H}_{S_2}|_k$};
\draw [thick, ->] (0.6,-4.2) -- (2.8,-4.2);
\node at (1.9,-3.9) {$\mathcal{P}_U$};
\end{tikzpicture}
\end{gathered}
\end{equation}

and

\begin{equation} \label{PQ-TQFT}
\begin{gathered}
\begin{tikzpicture}[every tqft/.append style={transform shape}]
\draw[rotate=90] (0,0) ellipse (2.5cm and 1 cm);
\node[above] at (0,1.7) {$\mathscr{B}$};
\node at (-0.5,1.1) {$P$};
\node at (-0.5,-1.1) {$R$};
\begin{scope}[tqft/every boundary component/.style={draw,fill=green,fill opacity=1}]
\begin{scope}[tqft/cobordism/.style={draw}]
\begin{scope}[rotate=90]
\pic[tqft/reverse pair of pants, at={(-1,0)}, name=a];
\pic[tqft/pair of pants, anchor=incoming boundary 1, name=b, at=(a-outgoing boundary 1)];
\end{scope}
\end{scope}
\end{scope}
\draw[rotate=90] (0,-4) ellipse (2.5cm and 1 cm);
\node[above] at (4,1.7) {$\mathscr{B}$};
\node at (2,0.7) {$\mathscr{S}$};
\node at (4.4,1.1) {$Q$};
\node at (4.4,-1.1) {$R$};
\draw [thick, ->] (0,-2.7) -- (0,-3.8);
\draw [thick, ->] (4,-2.7) -- (4,-3.8);
\node at (0,-4.2) {$\mathcal{H}_P|_i \bigotimes \mathcal{H}_{R}|_i$};
\node at (4.2,-4.2) {$\mathcal{H}_{Q}|_k \bigotimes \mathcal{H}_{R}|_k$};
\draw [thick, ->] (1.2,-4.2) -- (2.8,-4.2);
\node at (1.9,-3.9) {$\mathcal{P}_U$};
\end{tikzpicture}
\end{gathered}
\end{equation}
\noindent
define TQFTs on an observed system $S$ between ``times'' $i$ and $k$.  Combining Diagrams \eqref{CCCD-to-Cob} and \eqref{S-TQFT}, and also \eqref{CCCD-to-Cob-2} and \eqref{PQ-TQFT}, yield the required map from QRFs to TQFTs.

As any CCCD represents a finite computation, any CCCD can, in principle, be implemented by a finite physical device, e.g. a physical implementation of a finite Turing Machine --- otherwise, by possible attachment to a topological version of the latter \cite{merelli:19}.  Hence a QRF can, in principle, be constructed given any finite cobordism.  Whether such a QRF could be implemented as a practical measurement device remains, in general, an open question.

\section{Mapping sequential measurements to TQNNs}\label{CCCD-TQNN}

We may consider a $m$-dimensional sector $X$ of some (boundary) space $\mathscr{B}$, assumed to have either discrete or a discretized topology and dimension $N \ll m$. We then consider sectors of $\mathscr{B}$, for instance the two sets $S_1$ and $S_2$, and $P$, $Q$ and $R$, and endow these latter ones with linear maps of the type considered in \eqref{hilbert-map-1}. In general, we can consider a generic set of sectors $\{ S_i\}$, with $i=1, \dots p$ such that $p\leq m$, and hence write the maps
\begin{equation}\label{sector-map-p}
S \rightarrow  {\times}_{ \{  i_1 \} } {S}_{ \{  i_1  \} }  \rightarrow  \times_{ \{  i_2 \} } S_{ \{  i_2  \} }  \rightarrow  \dots  \times_{ \{  i_p \} } S_{ \{  i_p  \} }  \rightarrow  S
\end{equation}
for any compatible $\{  i_p  \}$-decomposition in sectors of $\mathscr{B}$, with $p\leq m$.
We further assume that the decomposition fulfils the structure of respective objects and arrows in the category $\mathbf{CCCD}$. Generalizing the linear map \eqref{hilbert-map-1} so to be consistent with \eqref{sector-map-p}, we can write
\begin{equation}\label{hilbert-map-p}
\mathcal{H}_S \rightarrow  \bigotimes_{\{  i_1  \}} \mathcal{H}_{S_{\{  i_1 \}}} \rightarrow  \bigotimes_{\{  i_2  \}} \mathcal{H}_{S_{\{  i_2 \}}} \rightarrow \cdots  \bigotimes_{\{  i_p  \}} \mathcal{H}_{S_{\{  i_p \}}} \rightarrow  \mathcal{H}_S
\end{equation}

The discrete topology of $\mathscr{B}$, and consequently its sectors, supports quantum states belonging to the Hilbert spaces of each sector in the decomposition.

In particular, we consider in this section quantum states that are supported on graphs $\gamma$, the latter ones being the union of path $\gamma_i$, with some $i=1,\dots n$, at nodes $v$. These quantum states are often refereed to in the literature as spin-networks, which were first introduced by Penrose --- see e.g. \cite{penrose:71} --- and the relevance of which for topological quantum field theories and quantum gravity has been pointed out by Rovelli and Smolin \cite{rovelli:95}.

\subsection{$G$-bundle structures and emergent metric} \label{g-bundle-metric}

TQFTs can be taken into account that are locally gauge-invariant under the action of a Lie group $G$, the Lie algebra of which $\mathfrak{g}$ is equipped with an invariant non-degenerate bilinear form $\langle \cdot , \cdot \rangle$.

We can consider, for the sake of simplicity, the formulation of a generic $BF$ theory \cite{baez:99} with local gauge group $G$. Taking into account a smooth $d$-dimensional manifold $\mathcal{M}$, it is possible to construct a principal $G$-bundle $M$ over $\mathcal{M}_d$ provided that a $G$-connection $A$ over the bundle $M$ is introduced, and a ``frame field'' that is an \ad(M)-valued $(d-2)$-form $B$ over $\mathcal{M}_d$, where \ad(M) is the vector bundle associated to $M$ by means of the adjoint action of $G$ on its Lie algebra. Picking a local trivialization on $\mathcal{M}_d$, one can denote the curvature of $A$, represented as a $\mathfrak{g}$-valued 1-form, with a $\mathfrak{g}$-valued 2-form $F$. Hence $B$ must be intended as a $\mathfrak{g}$-valued (d-2)-form in order to write consistently over $\mathcal{M}_d$ the Lagrangian density of the theory:
\begin{equation}\label{BF}
\mathcal{L}= \langle B \wedge F \rangle \,.
\end{equation}

The classical (``kinematical'') phase-space of the theory is the cotangent bundle $T^*\mathcal{A}$. Here $\mathcal{A}$ stands for the space of connections on the restriction $M \vert_{S}$ of the bundle $M$ to an initial (in time) slice $\{0\} \times S$. A point of the phase-space then encodes a connection $A$ on $M \vert_{S}$ and an \ad($M \vert_{S}$)-valued $(d-2)$-form B on $S$. The symplectic structure is provided by the expression:
\begin{equation}
\omega((\delta A, \delta B), (\delta A', \delta B') )= \int _S \langle \delta A \wedge \delta B'- \delta A' \wedge \delta B \rangle\,.
\end{equation}

Variation of \eqref{BF} yields:
\begin{equation}\label{EOM1}
0=\delta \int_{\mathcal{M}_d} \!\!\! \mathcal{L}= \int_{\mathcal{M}_d} \!\!\! \langle \delta B \wedge F  + (-1)^d d_A B \wedge \delta A  \rangle\,,
\end{equation}
from which we can read the equation of motions of the theory:
\begin{equation}\label{EOM2}
 F=0\,,  \qquad \qquad   d_A B=0\,,
 \,
\end{equation}
having denoted with $d_A$ the exterior covariant derivative with respect to the connection $A$.

Equations \eqref{EOM2} tell us that the theory carries no local degrees of freedom, namely that the theory is topological. $F=0$ imposes the flatness of the connection $A$. Flat connections are then equivalent up to gauge transformations. On the other hand, all the solutions to the equation $d_A B=0$, often called the Gau\ss~constraints, are the same up to the a second class of gauge transformations, namely:
\begin{equation}\label{BFG}
 A \rightarrow  A \,,  \qquad \qquad   B \rightarrow B + d_A \eta \,,
\end{equation}
with \ad(P)-valued $(d-3)$-form $\eta$, which leave invariant the action \eqref{BF}.

The field equations \eqref{EOM2} constrain the initial data for $A$ and $B$ and realize a symplectic reduction of $T^*\mathcal{A}$ to the physical phase space. The Gau\ss~constraint  $d_A B=0$ generates indeed gauge transformations on $T^*\mathcal{A}$. If we denote with $\mathcal{G}$ the group of the gauge transformations of the bundle $M \vert_S$, then the ``gauge-invariant phase space'' can be denoted as  $T^*(\mathcal{A}/\mathcal{G})$.
The condition $F=0$, imposing flatness of the connection $A$, generates transformations of the form \eqref{BFG}, which are peculiar of the $BF$ theory. Denoting with $\mathcal{A}_0$ the space of flat connections on $M \vert_S$, the symplectic reduction operated by the constraints \eqref{EOM2} individuates the physical phase space of the theory, namely $T^*(\mathcal{A}_0/\mathcal{G})$.

The theory in \eqref{BF} can be canonically quantized by constructing the Hilbert space of the square integrable functionals over $\mathcal{A}$, and then imposing the constraints \eqref{EOM2} at the quantum level. The Hilbert space of square-integrable functionals on the configuration space $\mathcal{A}$, denoted with $L^2(\mathcal{A})$, is the algebra $Fun(\mathcal{A})$ of functional on $\mathcal{A}$ of the form:
\begin{equation}
\Psi(A)=f(H_{\gamma_1}(A), \dots , H_{\gamma_n}(A))\,,
\end{equation}
with $f$ a continuous complex-valued function of finitely many holonomies:
\begin{equation}
H_{\gamma_i}(A)= P \exp(\int_{\gamma_i} A)
\end{equation}
on a generic path $\gamma_i \in S$, for $i=1, \dots n$ and with $P$ denoting path-ordering in the exponential. Holonomies $H_{\gamma_i}(A)$ along the paths $\gamma_i$ are elements of $G$. Thus $Fun(\mathcal{A})$ is isomorphic to the algebra of all the possible continuous complex-valued functionals on $G^n$.

An inner product on $Fun(\mathcal{A})$ can be defined, in order to complete the definition of the Hilbert space $L^2(\mathcal{A})$, introducing the notion of a graph, i.e. 1-complex, $\Gamma \in A$.

\begin{definition} \label{graph}
A {\rm graph} in $S$ is composed of a finite collection of real-analytic paths $\gamma_i: [0,1] \rightarrow S$ that are embedded and intersect only at their endpoints. We call $\gamma_i$ the {\rm links} of the graphs, and the their endpoints the {\rm nodes}. An orientation on the graph can be provided if, given a node $n$, a link is said to be {\rm outgoing} from $n$ when $\gamma_i(0)=n$, and is said to be {\rm incoming} when $\gamma_i(1)=n$.
\end{definition}

\noindent
Then the inner product is simply the integral:
\begin{equation}\label{inner-product}
\langle \Psi , \Phi \rangle = \int_{G^n} \overline{\Psi}  \Phi\,
\end{equation}
over the Haar measure on $G^n$.

Completing the algebra $Fun(\mathcal{A})$ with the inner product \eqref{inner-product} one obtains the kinematical Hilbert space $L^2(\mathcal{A})$, over which the constraints \eqref{EOM2} act at the quantum level.

The action of the Gau\ss~constraint $d_A B=0$ can be easily implemented. Indeed, taking into account gauge transformations we can first define $Fun(\mathcal{A}/\mathcal{G})$ as the space of all functionals in $Fun(\mathcal{A})$ that are gauge-invariant. Completing the algebra $Fun(\mathcal{A}/\mathcal{G})$ with the norm defined in \eqref{inner-product} provides the gauge-invariant Hilbert space $L^2(\mathcal{A}/\mathcal{G})$.

The curvature constraint $F=0$ of \eqref{EOM2} must be still imposed at the quantum level. In order to proceed with this next step, we recall that the Hilbert space $L^2(\mathcal{A}/\mathcal{G})$ is spanned by the spin-network states.

\begin{definition} \label{spin-network}
A {\rm spin-network} in $S$ is a triple $\Psi=(\gamma, \rho, \iota)$ that consists of: 1) a graph $\gamma\in S$; 2) for each link $\gamma_i\in \gamma$, an irreducible representation $\rho_l$ of $G$; 3) for each node $n$, an intertwining operator, also called intertwiner, $\iota_n$ such that $\iota_n: \rho_{l_1} \otimes \dots \otimes  \rho_{l_n} \rightarrow \rho_{l_1'} \otimes \dots \otimes  \rho_{l_n'}$, with $l_1, \dots l_n$ links incoming to $n$ and $l_1', \dots l_n'$ links outgoing from $n$.
\end{definition}

From a perspective borrowed from gravitational systems described by the Einstein-Hilbert action, which in lower than $4$-dimensions can be cast as $BF$ theories, the imposition of the curvature constraint $F=0$ generates transformations in the form of diffeomorphisms, and thus implements the dynamics of the boundary states, as described by spin-network states on the sections $S$. In \cite{noui:2004} it was shown how to make sense of the formal expression for the physical projector $P$, namely:
\begin{equation} \label{phys-pro}
P=\int \mathcal{D} N \exp (\imath \int_S \langle N \, \widehat{F}\rangle  )
\,,
\end{equation}
with $N$ an {\it ad}(P)-valued Lagrangian multiplier and $\imath=\sqrt{-1}$ imaginary unit, and then implement the curvature constraint $F=0$ in the (physical) inner products of spin-network states, i.e.
\begin{equation} \label{inn-pro-phys}
\langle \Psi , \Phi \rangle_{\rm phys} = \langle P \Psi , \Phi \rangle\,,
\end{equation}

Generalizing the same argument to a topological $BF$ theory in a $d$-dimensional manifold, we can consider a local patch $\Sigma \in S$, with cellular decomposition made of squares of infinitesimal coordinates length $\delta$, and then regularize the curvature constraint as:
\begin{equation}
F[N]= \int_\Sigma \langle N F(A) \rangle = \lim_{\delta \rightarrow 0} \sum_{p^j} \delta^2 \langle N_{p^j} F_{p^j} \rangle\,,
\end{equation}
in which ${p^j}$ labels the $j$-th plaquette and $N_{p^j}$ is the discretization of the \ad(P)-valued Lagrangian multiplier $N$ evaluated at an interior point of the plaquette ${p^j}$. It is then relevant to remind that the discretized version of the holonomy $H_{p^j}[A]$ around the plaquette ${p^j}$ is a $G$ group-element that can be written as:
\begin{equation}
H_{p^j}[A]= 1\!\!1 + \delta^2 F_{p^j}(A) + \mathcal{O}(\delta^2)
\,,
\end{equation}
implying:
\begin{equation} \label{smeared-curve}
F[N]= \int_\Sigma \langle N F(A) \rangle = \lim_{\delta \rightarrow 0} \sum_{p^j} \delta^2 \langle N_{p^j} H_{p^j}[A] \rangle\,.
\end{equation}
This finally defines the regularized expression for the action of the  physical projection operator entering the physical scalar product of spin-network states:
\begin{equation} \label{reg-inn-pro}
\langle \Psi , \Phi \rangle_{\rm phys}= \lim_{\delta \rightarrow 0} \langle \prod_{p^{j}} \int \mathcal{D} N_{p^{j}} \exp(\iota \langle
N_{p^{j}} \widehat{H}_{p^{j}} \rangle)  \Psi, \Phi  \rangle =  \lim_{\delta \rightarrow 0} \langle \prod_{p^{j}} \delta(H_{p^j}) \Psi, \Phi \rangle\,.
\end{equation}
We can recast the previous construction in a more precise way, as follows:

\begin{definition} \label{cellular-decomposition}
For any pair of spin-network states $\Psi$ and $\Phi$ that belong to $L^2(\mathcal{A}/\mathcal{G})$ and are supported on the graphs $\gamma_\Psi$ and $\gamma_\Phi$, one can define the {\rm union graph} $\gamma_{\Psi\Phi}$ that satisfies the properties: 1) $\gamma_\Psi, \gamma_\Phi \in \gamma_{\Psi\Phi}$; 2) $\gamma_{\Psi\Phi}$ is the 1-skeleton of a cellular decomposition $\Sigma_{\Psi\Phi}$ of $S$; 3) the 2-complex $\Sigma_{\Psi\Phi}$ has the minimal number of 2-cells; then the set of irreducible loops $\alpha_{\Psi\Phi}$ can be defined as the set of oriented boundaries of the corresponding 2-cells in $\Sigma_{\Psi\Phi}$.
\end{definition}

With this Definition \ref{cellular-decomposition}, we can state the following theorem:

\begin{theorem} \label{cell-dec-phys-inn-pro}
The physical inner product among two any states spin-network states $\Psi$ and $\Phi$ that belong to $L^2(\mathcal{A}/\mathcal{G})$ is provided by the expression:
\begin{equation} \label{dec-phys-inn-pro}
\langle P \Psi, \Phi \rangle = \lim_{\delta \rightarrow 0} \frac{1}{N_{\delta, p} ! } \sum_{\sigma(\{j\})} \langle  \prod_{p^{\sigma(j)}}
\delta(H_{p^{\sigma(j)}}) \Psi, \Phi \rangle =
\langle \prod_{\alpha\in \alpha_{\Psi\Phi}} \delta(H_\alpha) \Psi, \Phi \rangle \,,
\end{equation}
where $N_{\delta, p}$ is the $\delta$-dependent number of plaquettes, $\sigma(\{j\})$ denotes a permutation of the set of the plaquette labels and $H_{\alpha}$ denotes a holonomy around an irreducible loop $\alpha \in \alpha_{\Psi\Phi}$.
\end{theorem}

\begin{proof}
The proof follows from the previous arguments, in equations
\eqref{phys-pro}, \eqref{inn-pro-phys}, \eqref{smeared-curve} and
\eqref{reg-inn-pro}.
\end{proof}

The physical inner product defined in \eqref{reg-inn-pro} provides the amplitudes for the transition among states in the topological theory. Hence, it individuates the dynamics of the $BF$ theory.

A framework has been constructed in \cite{marciano:20} that associates spin-network states that are colored with irreducible representations of some Lie group $G$ to training and test sample states in (quantum) machine learning. A typical case-study is represented by training and test samples that are images composed of pixels. The set of pixels that compose these images can be associated to the discretization of some subset $\Sigma_2$ of a $2$-dimensional manifold $S_2$. Dual 1-skeleton graphs can be then constructed, with four-valent nodes positioned inside the pixels and links interconnecting the nodes. A discrete or continuous Lie group, or eventually a non-trivial Hopf algebra, can be then employed to provide the colouring of the graphs. This choice is dependent on the features of the set of pixels. For instance, if the pixels deploy eleven different shades of grey, including the colours white and black, one can assign to this set of shades of grey the weights $m=\{0,1,\dots 10 \} \in\mathbb{N}$. To each link connecting the nodes of adjacent pixels sharing boundaries, we can then assign the minimum of the weights associated to each adjacent pixel. In other words, given the node $n_1$ and its adjacent node $n_2$, we assign to the link $l_{12}$ interconnecting them the weight $m_{12}={\rm min}\{m_1,m_2\}$. If the samples show a rotational invariance, one can implement the Lie group $\SO(3)$, or its double covering $\SU(2)$, as the local gauge symmetry of the topological theory. If we pick up a $3$-dimensional $BF$ theory with $G=\SU(2)$ local symmetry, we can assign to the link $l_{ij}$ interconnecting the nodes $n_i$ and $n_j$ of the 1-skeleton graph associated to the pixellized image irreducible representations (irreps) of $\SU(2)$ in the form $j_{ij}=m_{ij}/2$. These can be depicted as $m_{ij}$ strands in the fundamental representations that are symmetrized according to the Jones-Wenzl projector \cite{kauffman:94}. At the same time, irreps must be compatible at each node. Therefore, the unpaired strands at the nodes shall be closed in loops and ``removed'' from the irreps assignments to each link. In this way we can obtain irreps that are compatible at each node, and finally consider to saturate the representation indices with those ones of compatible  intertwiner tensors at the nodes. The resulting construction provides, for each element of either the training or the test sample, a state of a topological quantum neural network (TQNN), namely a spin-network state that is an element of $L^2( \mathcal{A}/\mathcal{G})$ and evolves according to the projector of a $\SU(2)$-symmetric $3$-dimensional $BF$ theory, as implemented in the inner product \eqref{dec-phys-inn-pro}. Therefore, the physical inner product provides a cobordism $\mathscr{S}$ of gauge-invariant Hilbert spaces spanned by the spin-networks basis. In particular, the evolution $\mathscr{S}$ instantiates a classifier $\mathscr{I}$, when states of the input gauge-invariant Hilbert space are confronted with archetypes of the output gauge-invariant Hilbert space. States of these two boundary Hilbert spaces are the input and output states of the TQNN, respectively, and their functorial evolution by physical inner product encodes a topological data structure, which we refer to as {\em a Topological Quantum Neural 2-Complex (TQN2C).}

We can generalize this construction to any $d$-dimensional manifold $\mathcal{M}_d$ and any group $G$, introducing  TQNNs as it follows.

\begin{definition} \label{topological-quantum-neural-networks}
{\rm Topological quantum neural networks (TQNNs)} comprise elements of the training and test samples in frameworks of quantum machine learning that are represented by spin-network states of some gauge-invariant Hilbert space $L^2(\mathcal{A}/\mathcal{G})$ and that evolve according to a $BF$ theory over a $G$-bundle and $d$-dimensional base manifold $\mathcal{M}_d$. The cobordism $\mathscr{S}$ instantiated by the physical inner product in the $BF$ theory realizes a classifier $\mathscr{I}$ for these TQNN elements, which implements a topological quantum neural 2-complex (TQN2C).
\end{definition}

Diagrammatically, the Identity map $\mathscr{B}=\mathcal{M}_d \rightarrow \mathscr{B}=\mathcal{M}_d$ is implemented by a unitary propagator $\mathcal{P}_U$ represented by the action of the physical projector $P$.  Hence a classifier between TQNN elements can be represented as it follows:

\begin{equation} \label{full-TQN2C}
\begin{gathered}
\begin{tikzpicture}
\draw[rotate=90] (0,0) ellipse (2.5cm and 1 cm);
\node[above] at (0,1.7) {$\mathscr{B}$};
\draw[rotate=90] (0,-4) ellipse (2.5cm and 1 cm);
\node[above] at (4,1.7) {$\mathscr{B}$};
\node at (2,2.8) {$\mathscr{S}$};
\draw [thick] (0,2.5) -- (4,2.5);
\draw [thick] (0,-2.5) -- (4,-2.5);
\draw [thick, ->] (0,-2.7) -- (0,-3.7);
\draw [thick, ->] (4,-2.7) -- (4,-3.7);
\node at (-0.1,-4) {$\mathcal{H}_A|_i$};
\node at (3.9,-4) {$\mathcal{H}_A|_k$};
\draw [thick, ->] (0.5,-4) -- (3.3,-4);
\node at (1.9,-3.7) {$\mathcal{P}_U$};
\end{tikzpicture}
\end{gathered}
\end{equation}
\noindent
that maps the cobordism $\mathscr{S}=\mathscr{I}$ of $\mathscr{B}=\mathcal{M}_d$ to $A$'s apparent Hilbert space $\mathcal{H}_A$ describing bit-string encodings on $\mathscr{B}=\mathcal{M}_d$, as measured at ``times'' input $i$ and output $k$.  Hence the proof of Theorem \ref{thm2} itself defines a TQN2C.  Similarly, the diagrams:

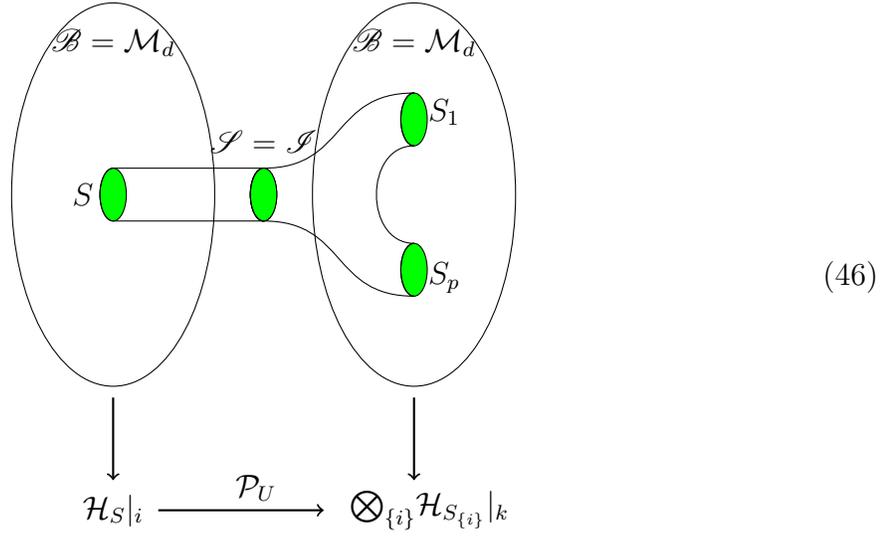
\begin{equation} \label{S-TQN2C}
\begin{gathered}
\begin{tikzpicture}[every tqft/.append style={transform shape}]
\draw[rotate=90] (0,0) ellipse (2.55cm and 1.35 cm);
\node[above] at (0,1.7) {$\mathscr{B}=\mathcal{M}_d$};
\node at (-0.4,0) {$S$};
\begin{scope}[tqft/every boundary component/.style={draw,fill=green,fill opacity=1}]
\begin{scope}[tqft/cobordism/.style={draw}]
\begin{scope}[rotate=90]
\pic[tqft/cylinder, name=a];
\pic[tqft/pair of pants, anchor=incoming boundary 1, name=b, at=(a-outgoing boundary 1)];
\end{scope}
\end{scope}
\end{scope}
\draw[rotate=90] (0,-4) ellipse (2.55cm and 1.35 cm);
\node[above] at (4,1.7) {$\mathscr{B}=\mathcal{M}_d$};
\node at (2,0.7) {$\mathscr{S}=\mathscr{I}$};
\node at (4.4,1.1) {$S_1$};
\node at (4.4,-1.1) {$S_p$};
\draw [thick, ->] (0,-2.7) -- (0,-3.8);
\draw [thick, ->] (4,-2.7) -- (4,-3.8);
\node at (0,-4.2) {$\mathcal{H}_S|_i$};
\node at (4.2,-4.2) {$ {\bigotimes}_{ \{ i \} } \mathcal{H}_{S_{ \{ i \} } }|_k$};
\draw [thick, ->] (0.6,-4.2) -- (2.8,-4.2);
\node at (1.9,-3.9) {$\mathcal{P}_U$};
\end{tikzpicture}
\end{gathered}
\end{equation}

and

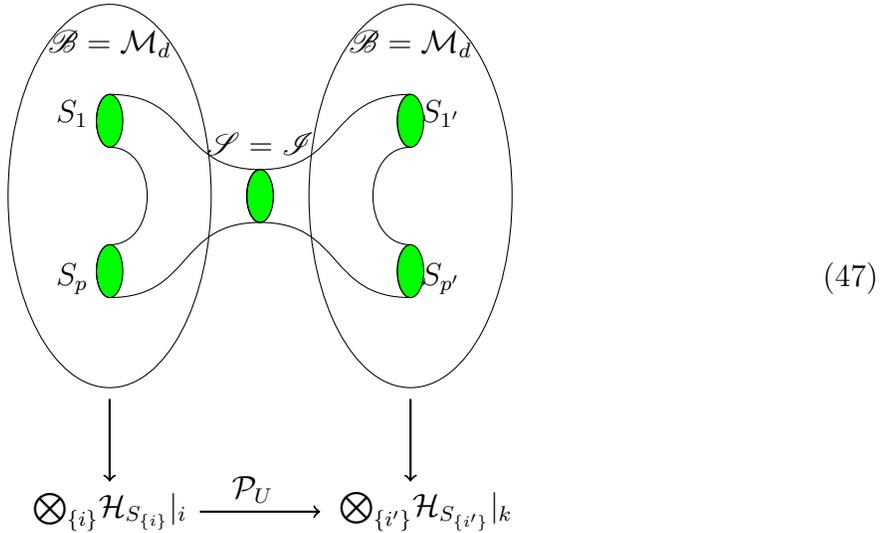
\begin{equation} \label{PQ-TQN2C}
\begin{gathered}
\begin{tikzpicture}[every tqft/.append style={transform shape}]
\draw[rotate=90] (0,0) ellipse (2.55cm and 1.35 cm);
\node[above] at (0,1.7) {$\mathscr{B}=\mathcal{M}_d$};
\node at (-0.5,1.1) {$S_1$};
\node at (-0.5,-1.1) {$S_p$};
\begin{scope}[tqft/every boundary component/.style={draw,fill=green,fill opacity=1}]
\begin{scope}[tqft/cobordism/.style={draw}]
\begin{scope}[rotate=90]
\pic[tqft/reverse pair of pants, at={(-1,0)}, name=a];
\pic[tqft/pair of pants, anchor=incoming boundary 1, name=b, at=(a-outgoing boundary 1)];
\end{scope}
\end{scope}
\end{scope}
\draw[rotate=90] (0,-4) ellipse (2.55cm and 1.35 cm);
\node[above] at (4,1.7) {$\mathscr{B}=\mathcal{M}_d$};
\node at (2,0.7) {$\mathscr{S}=\mathscr{I}$};
\node at (4.4,1.1) {$S_{1'}$};
\node at (4.4,-1.1) {$S_{p'}$};
\draw [thick, ->] (0,-2.7) -- (0,-3.8);
\draw [thick, ->] (4,-2.7) -- (4,-3.8);
\node at (0,-4.2) {$ {\bigotimes}_{ \{ i \} } \mathcal{H}_{S_{ \{ i \} } }|_i$};
\node at (4.2,-4.2) {$ {\bigotimes}_{ \{ i' \} } \mathcal{H}_{S_{ \{ i' \} } }|_k$};
\draw [thick, ->] (1.2,-4.2) -- (2.8,-4.2);
\node at (1.9,-3.9) {$\mathcal{P}_U$};
\end{tikzpicture}
\end{gathered}
\end{equation}
\noindent
define two distinct TQN2Cs on an observed system $S\in \mathcal{M}_d$ between ``times'' input $i$ and output $k$.  Combining Diagrams \eqref{CCCD-to-Cob} and \eqref{S-TQN2C}, and also \eqref{CCCD-to-Cob-2} and \eqref{PQ-TQN2C}, yields the map from QRFs to TQN2Cs.

Through this construction, the sectors $S$ and their TQNNs can be assimilated to the level of the individual qubits on the finite boundary $\mathcal{B}$, which is the screen of the holographic projection of data and information.

\subsection{Spin manifolds and the index theorem}

The topological features of TQNNs enable information transfer and hence pattern recognition among boundary states of their TQN2C functorial evolutions. As we have shown previously, these boundary states are called training and test samples in the jargon of machine learning. Since the support of TQNNs and TQN2C, represented respectively by 1-complexes and 2-complexes, is dual to the discretization of the boundary manifold $\mathscr{B}$ and the bulk manifold $\mathcal{M}$, it is relevant to recall how the topological features of these latter manifolds can be characterized in terms of the operators entering the definition of the theories invoked to establish the classifier amplitudes. This task can be accomplished deploying the powerful tools provided by the Atiyah-Singer theorem \cite{atiyah:63,atiyah:68a,atiyah:68b,atiyah:68c,atiyah:69,atiyah:71a,atiyah:71b,atiyah:84}, which extend to elliptic and skew-adjoint Fredholm operators.

We proceed to recall several distinctive examples of locally symmetric structures for the base manifolds $\mathcal{M}_d$, and consider the cases when these are endowed with a metric. We first focus on metric $d$-dimensional manifolds $\mathcal{M}_d$ that are oriented and Riemannian. The presence of the metric, which for instance can been induced by considering the action \eqref{Pleb}, as we will comment in section \ref{gauge-inv}, enables reduction of the internal symmetry group structure to a spin structure.
Without loss of generality, we may focus on the $d=4$ dimensional case. A point of a Euclidean $4$-dimensional (flat) space manifold $\mathcal{M}_4$, for which the relevant group is $\Spin(4)$, can be thought as the real vector space $V$ of quaternions, represented by $2\times 2$ matrices of the form:
\begin{eqnarray}
Q = \left(\begin{array}{cc} t + \imath z &  -x +\imath y \\ x + \imath y & t - \imath z \end{array}\right)\,,
\end{eqnarray}
with $t$, $x$, $y$ and $z$ real coordinates of the generic point of $\mathcal{M}_4$. $Q$ is generated by the four $2\times 2$ matrices:
\begin{eqnarray}
{\bf 1}= \left(\begin{array}{cc} 1 & 0 \\ 0 & 1 \end{array}\right)\,, \quad
{\bf i}= \left(\begin{array}{cc} 0 & -1 \\ 1 & 0 \end{array}\right)\,,
\quad
{\bf j}= \left(\begin{array}{cc} 0 & \imath \\ \imath & 0 \end{array}\right)\,,
\quad
{\bf k}= \left(\begin{array}{cc} \imath & 0 \\ 0 & -\imath \end{array}\right)\,,
\end{eqnarray}
so that it casts $Q=t {\bf 1} + x {\bf i} + y {\bf j}+ z{\bf k}$, and has determinant that individuates the line element of the Euclidean manifold $\mathcal{M}_4$, namely:
\begin{eqnarray}
\Det~Q = t^2 + x^2 + y^2 + z^2\,.
\end{eqnarray}
(The unit sphere on the Euclidean manifold $\mathcal{M}_4$ can be for instance identified in terms of element of $\SU(2)=\{Q\in V: \Det(Q)=1 \}$.) The group of the isomorphisms from $V$ to the linear group $GL(V)$, is the representation of the compact Lie group $\Spin(4)$, which is the direct product of two copies of $\SU(2)$, namely $\Spin(4)= \SU_+(2) \times \SU_-(2)$. Denoting the elements of the decomposition of $\Spin(4)$ with $g_{\pm}\in \SU_{\pm}(2)$, the representation $\rho$ of the action of $\Spin(4)$ on $Q$ is expressed by:
\begin{eqnarray}
\rho(g_+, g_-)(Q)= g_- Q g_+\,.
\end{eqnarray}
Using the Euclidean metric structure of $\mathcal{M}_4$, provided by the determinant of $Q$, one easily recognizes that:
\begin{eqnarray}
\Det(g_- Q g_+)= \Det ~Q  \,,
\end{eqnarray}
i.e. the representation $\rho$ preserves the Euclidean inner product, thus ensuring the reduction:
\begin{eqnarray}
\rho: \Spin(4) \rightarrow \SO(4) \in \GL(V).
\end{eqnarray}

A similar construction can be provided for a Lorentzian spacetime manifold, with relevant group $\SL(2,\mathbb{R})$. In this case, points of a (flat) Lorentzian manifold $\mathcal{M}_4$ can be thought as a $2 \times 2$ Hermitian matrix:
\begin{eqnarray}
X= \left(\begin{array}{cc} t + z &  x -\imath y \\ x + \imath y & t - z \end{array}\right)= t \, 1\!\!1 + x\,\sigma_x + y\,\sigma_y + z\,\sigma_z \,,
\end{eqnarray}
with $1\!\!1$ identity and $\sigma_x$, $\sigma_y$ and $\sigma_z$ Pauli matrices.  The action by a group element $g\in \SL(2,\mathbb{R})$ on $X$ is defined by:
\begin{eqnarray}
X \rightarrow g X^t \bar{g}\,,
\end{eqnarray}
which leaves invariant the line element on  $\mathcal{M}_4$, namely $\Det(g X^t \bar{g})= \Det~X$. This implies that $\SL(2,\mathbb{R})$ covers the identity in the Lorentz group.

Coming back to a manifold $\mathcal{M}_4$ endowed with Euclidean metric, we pick now the internal symmetry structure of the tangent bundle $T \mathcal{M}_4$ of $\mathcal{M}_4$ to be $\SO(4)$, with reduction from the $\GL(4,\mathbb{R})$ being realized by the Riemannian metric. Choosing for $T \mathcal{M}_4$ a trivializing open cover  $\{ U_\alpha: \alpha \in A \}$, with index $\alpha\in A$, we can write the transition functions that take values to $\SO(4)$ as $g_{\alpha \beta}: U_\alpha \cap U_\beta \rightarrow SO(4)\subset \GL(4,\mathbb{R})$.

\begin{definition} \label{spin-structure}
A {\rm spin structure} on $\mathcal{M}_4$ is provided by a open covering $\{ U_\alpha: \alpha \in A \}$ of  $\mathcal{M}_4$  and by the transition functions $\tilde{g}_{\alpha \beta}: U_\alpha \cap U_\beta \rightarrow \Spin(4)$, which are such that $\rho \circ \tilde{g}_{\alpha \beta} = g_{\alpha \beta}$ and the cocycle condition holds, namely  $\tilde{g}_{\alpha \beta} \tilde{g}_{\beta \gamma}=\tilde{g}_{\alpha \gamma}$ on $U_\alpha \cap U_\beta \cap U_\gamma$.
\end{definition}

A manifold endowed with a spin structure is called a spin manifold.
We can further extend Definition \ref{spin-structure} in order to encode a $\Spin(4)^c$ structure.

\begin{definition} \label{spin-c-structure}
A {\rm spin$\!\!\!\phantom{a}^c$ structure} on $\mathcal{M}_4$ is provided by a open covering $\{ U_\alpha: \alpha \in A \}$ of  $\mathcal{M}_4$  and by the transition functions $\tilde{g}_{\alpha \beta}: U_\alpha \cap U_\beta \rightarrow \Spin(4)\!\!\!\phantom{a}^c$, which are such that $\rho \circ \tilde{g}_{\alpha \beta} = g_{\alpha \beta}$ and the cocycle condition holds.
\end{definition}

The conditions for $\mathcal{M}_4$ to posses a spin$\!\!\!\phantom{a}^c$ structure are that: i) there exist transition functions $\tilde{g}_{\alpha \beta}: U_\alpha \cap U_\beta \rightarrow Spin(4)$; ii) there exists a complex line bundle $L$ over $\mathcal{M}_4$ with hermitian metric and transitions functions $h_{\alpha \beta}: U_\alpha \cap U_\beta \rightarrow U(1)$. Then, one can recover a spin$\!\!\!\phantom{a}^c$ structure on $\mathcal{M}_4$ by considering the transition functions $h_{\alpha \beta} \tilde{g}_{\alpha \beta}: U_\alpha \cap U_\beta \rightarrow Spin(4)\!\!\!\phantom{a}^c$.

Without proving it, we enunciate the following theorem.

\begin{theorem}[\cite{lawson:89}] \label{compact-oriented-manifold}
Every compact oriented $4$-dimensional manifold has a spin$\!\!\!\phantom{a}^c$ structure.
\end{theorem}

Once a spin$\!\!\!\phantom{a}^c$ structure is provided through $\tilde{g}_{\alpha \beta}: U_\alpha \cap U_\beta \rightarrow \Spin(4)\!\!\!\phantom{a}^c$, one can introduce the transition functions $\pi \circ \tilde{g}_{\alpha \beta}: U_\alpha \cap U_\beta \rightarrow  {\rm U}(1)$, which determine a complex line bundle denoted with $L^2$.

Furthermore, since the transition functions satisfy the cocycle condition, one can naturally introduce vector bundles. Assuming indeed that $\mathcal{M}_4$ has a spin structure induced by $\tilde{g}_{\alpha \beta}: U_\alpha \cap U_\beta \rightarrow \Spin(4)$, then two vector bundles of rank 2 over $\mathcal{M}_4$, respectively $W_+$ and $W_-$, can be introduced by the transition functions:
\begin{equation}
\rho^c_+ \circ \tilde{g}_{\alpha \beta}: U_\alpha \cap U_\beta \rightarrow \SU_+(2)\,, \qquad \rho^c_- \circ \tilde{g}_{\alpha \beta}: U_\alpha \cap U_\beta \rightarrow \SU_-(2)\,.
\end{equation}
Thanks to the spin structure on $\mathcal{M}_4$, one can represent the complexified tangent bundle on  $\mathcal{M}_4$ in terms of two complex vector bundles that are more basic, namely
$T \mathcal{M}_4 \otimes \mathbb{C} \cong \Hom_{\mathbb{C}}(W_+,W_-)$. The bundles $W_+$ and $W_-$ can also be regarded as quaternionic line bundles over $\mathcal{M}_4$. Once the complex line bundle $L$ over $\mathcal{M}_4$ is also introduced, one obtains:
\begin{equation} \label{Hom}
T \mathcal{M}_4 \otimes \mathbb{C} \cong \Hom_{\mathbb{C}}(W_+ \otimes  L,W_-  \otimes  L)\,.
\end{equation}

The representations $\rho^c_+$ and $\rho^c_-$ induce the ${\rm U}(2)$-bundles $W_+\otimes L$ and $W_-\otimes L$. Similar to the case of spin structures, spin$\!\!\!\phantom{a}^c$ structures enable to represent a complexified tangent bundle in terms of the more basic bundles $W_+\otimes L$ and $W_-\otimes L$, according to \eqref{Hom}. The sections of $W_+\otimes L$ and $W_-\otimes L$ are respectively spinor fields of positive and negative chirality.

A spin connection can be introduced for a Riemannian manifold $\mathcal{M}_4$. Given the fact that the tangent bundle $T \mathcal{M}_4$ of $\mathcal{M}_4$ possesses as a canonical connection, the Levi-Civita connection, one can extend to $W_+$ and $W_-$ canonical connections inherited from the Levi-Civita connection. At this purpose, one can first construct a connection on the bundle $\End(W)$, which is easily recovered from the Levi-Civita connection on $T \mathcal{M}_4$, and hence introduce a connection on $W$ itself, a $\Spin(4)$-connection.

Finally, it is possible to prove that there exists a one-to-one correspondence between $\Spin(4)$-connections on $W$ and $\SO(4)$-connections on $TM$.

\begin{theorem}[\cite{lawson:89}] \label{connections}
If $\mathcal{M}_4$ is a $4$-dimensional oriented Riemannian manifold endowed with a spin structure, then there exists a unique $\Spin(4)$-connection on $W$, with $W=W_+ \oplus W_-$, which induces the Levi-Civita connection on $\End(W)$.
\end{theorem}

This theorem can be extended to a complex line bundle $L$ over a spin manifold $\mathcal{M}_4$, hence allowing us to define a connection on the bundle $W\otimes L$ by simply taking the tensor product of the connection on $W$ and of the connection on $L$.

Why we consider, as a foundation for the tensor algebra on a Riemannian manifold, the two vector bundles $W_+$ and $W_-$, alternatively $W_+\otimes L$ and $W_-\otimes L$, instead of the tangent bundle $T \mathcal{M}_4$, is that these structure are in fact more basic \cite{moore:01}. Furthermore, these structures provide the basis for introducing the Clifford algebra. To this extent we remind the reader that an element of the Euclidean space $V$ can be considered as a complex linear homomorphism from $W_+$ to $W_-$ and then represented by the quaternion $Q$ previously introduced. The construction can then be recast in terms of a skew-hermitian endomorphism of $W=W_+ \oplus W_-$, this in turn represented by a $2 \times 2$ quaternion matrix of the form:
\begin{equation}
\theta(Q)=
\left(\begin{array}{cc}
0 & -\bar(Q)^t \\
Q & 0
\end{array}
\right)\,,
\end{equation}
defining a complex linear map $\theta: V \otimes \mathbb{C} \rightarrow \End(W)$. $\End(W)$ is a $16$-dimensional algebra over $\mathbb{C}$, with algebra multiplication induced by the matrix multiplication and such that, for $Q\in V$,
\begin{equation} \label{com-clifford}
(\theta(Q))^2=
\left(\begin{array}{cc}
-\bar(Q)^t Q  & 0\\
0 & - Q^t \bar(Q)
\end{array}
\right)
=(- \Det~Q) 1\!\!1 \,.
\end{equation}
The composition in \eqref{com-clifford}, referred to as Clifford multiplication, enables to interpret $\End(W)$ as the complexification of the Clifford algebra of $V$, denoted $\Cl(V)$. One can also directly obtain the Clifford algebra $\Cl(V \otimes \mathbb{C})$ by constructing complex $4 \times 4$ matrices $\gamma_1$, $\gamma_2$, $\gamma_3$ and $\gamma_4$ fulfilling the product rule:
\begin{equation} \label{gamma-clifford}
\gamma_i \cdot \gamma_j + \gamma_j \cdot \gamma_i = - 2 \delta_{ij}\,,
\end{equation}
for $i,j=1,\dots 4$. The basis of the complex vector space $\End(W)$ is then seen to consist of matrices $1\!\!1_{4\times 4}$, $\gamma_i$, $\gamma_i \gamma_j$ for $i<j$, $\gamma_i \gamma_j \gamma_k$ for $i<j<k$ and $\gamma_1 \gamma_2 \gamma_3 \gamma_4$. One may then identify the complexified second exterior power $\Lambda^2 V \otimes \mathbb{C}$ with the complex subspace of $\End(W)$ generated by the $\gamma_i \gamma_j$ elements ($i<j$).

At the same time, one may denote by $\Lambda V \otimes \mathbb{C}$ the complexified first exterior power, and with $\Lambda^0 V \otimes \mathbb{C}$ the complex numbers field. The complexified third and fourth exterior powers $\Lambda^3 V \otimes \mathbb{C}$ and $\Lambda^4 V \otimes \mathbb{C}$ are respectively the complex subspaces of $\End(W)$ generated by $\gamma_i \gamma_j \gamma_k$  (for $i<j<k$) and $\gamma_1 \gamma_2 \gamma_3 \gamma_4$. This enables to write the decomposition of $\End(W)$ in its subspaces:
\begin{equation}
\End(W)= \Lambda^0 V \otimes \mathbb{C} \oplus \Lambda V \otimes \mathbb{C} \oplus \Lambda^2 V \otimes \mathbb{C}  \oplus \Lambda^3 V \otimes \mathbb{C}  \oplus \Lambda^4 V \otimes \mathbb{C}\,.
\end{equation}

Once a spin$\!\!\!\phantom{a}^c$ structure over a $4$-dimensional manifold $\mathcal{M}_4$ is introduced, a $\Spin(4)\!\!\!\phantom{a}^c$-connection can be also introduced and defined on the spin bundle $W\otimes L$, which is inherited from the canonical Levi-Civita connection on $T \mathcal{M}_4$.

This finally allows, while considering a $4$-dimensional Riemannian manifold $\mathcal{M}_4$ with spin$\!\!\!\phantom{a}^c$ structure and $\Spin(4)\!\!\!\phantom{a}^c$-connection $d_A$ on the spin bundle $W\otimes L$, to introduce the Dirac operator on the smooths section $\Gamma$ of the bundle.

\begin{definition} \label{dirac-operator}
The {\rm Dirac operator} $D_A: \Gamma(W\times L) \rightarrow \Gamma(W\times L)$ with coefficients in the line bundle $L$ is defined by the formula
$$
D_A(\psi)=\sum_{i=1}^4 \gamma_i \cdot d_A \psi(\gamma_i) = \sum_{i=1}^4 \gamma_i \cdot \nabla_{\gamma_i}^A \psi
\,.$$
\end{definition}

For a Euclidean $\mathcal{M}_4$ space, with global Euclidean coordinates $x_1$, $x_2$, $x_3$ and $x_4$, the bundle $W$ is trivial, and casts $W= \mathcal{M}_4 \times \mathbb{C}$, and $L$ is the trivial line bundle. Then the Dirac operator just recasts $D_A\psi = \sum_{i=1}^4 \gamma_i \frac{\partial \psi}{\partial x_i}$.

In its simplest form, the Dirac operator can be regarded as the operator $d+\delta$ of the Hodge theory, with square denoted $D^2_A$. The square of the Dirac operator, the Hodge Laplacian $\Delta=d \delta + \delta d$, is such that, given a form $\omega$, $(d+\delta)\omega =0$ if and only if $d\omega =0 $ and $\delta \omega =0 $, which in turn is true if and only if $\Delta \omega =0$. The kernel of $d+\delta$ is the finite-dimensional vector space of the harmonic forms. Denoting with $\Omega^p(\mathcal{M}_4)$ the space of $p$-forms over $\mathcal{M}_4$, once on $\mathcal{M}_4$ one has introduced $\Omega^+=\Omega^0(\mathcal{M}_4) \oplus \Omega^2(\mathcal{M}_4)  \oplus \Omega^4(\mathcal{M}_4) $ and $\Omega^-=\Omega^1(\mathcal{M}_4) \oplus \Omega^3(\mathcal{M}_4)$, the Hodge operator can be decomposed in two parts: $(d+\delta)^+: \Omega^+ \rightarrow \Omega^-$ and $(d+\delta)^-: \Omega^- \rightarrow \Omega^+$. One can finally define the index of $(d+\delta)^+$ as:
\begin{equation} \label{index}
{\rm Ind}(d+\delta)={\rm dim}({\rm Ker}(d+\delta)^+) - {\rm dim}({\rm Ker}(d+\delta)^-) \,.
\end{equation}

We denote with $\Xi^p(\mathcal{M}_4)$ the unique harmonic representative of the de Rham cohomology class of $p$-forms over $\mathcal{M}_4$, which is finite-dimensional and such that $\Omega^p(\mathcal{M}_4)=\Xi^p(\mathcal{M}_4) \oplus \Delta (\Omega^p(\mathcal{M}_4))=\Xi^p(\mathcal{M}_4)\oplus d(\Omega^{p-1}(\mathcal{M}_4))+ \delta(\Omega^{p+1}(\mathcal{M}_4))$, and with $\Xi^0(\mathcal{M}_4)$ the space of harmonic 0-forms, i.e. the space of constant functions.

We then remind that by definition the Hodge star $\star$ acting on $p$-forms over $\mathcal{M}_4$ fulfils $\star \Lambda^p V \rightarrow \Lambda^{4-p} V$, from which we derive the definition of co-differential $\delta= - \star d \, \star: \Omega^p(\mathcal{M}_4) \rightarrow \Omega^{p-1}(\mathcal{M}_4)$ and finally the aforementioned Hodge Laplacian $\Delta= d\delta +\delta d: \Omega^p(\mathcal{M}_4) \rightarrow \Omega^p(\mathcal{M}_4)$.

For $\omega$ a smooth $2$-form, one can define the actions of the operators $P_+$ and $P_-$, which split $\omega$ into:
\begin{equation}
\omega^+=P_+(\omega)=\frac{1}{2} (\omega+ \star \omega)\in \Omega_+^2(\mathcal{M}_4)\,,
\quad
\omega^-=P_-(\omega)=\frac{1}{2} (\omega- \star \omega)\in \Omega_-^2(\mathcal{M}_4)\,.
\end{equation}
These are, respectively, self-dual and anti-self-dual $2$-forms, sections on the bundles the fibers of which are $\Lambda_+^2 V$ and $\Lambda_-^2 V$. The Hodge $\star$ operator interchanges the kernels of the operators $d$ and $\delta$, ensuring that self-dual and anti-self-dual components of a harmonic $2$-form is still harmonic. The space of harmonic $2$-form on $\mathcal{M}_4$ hence decomposes into the direct sum decomposition:
\begin{equation}
\Xi^2(\mathcal{M}_4) \cong \Xi^2_+(\mathcal{M}_4)  \oplus \Xi^2_-(\mathcal{M}_4) \,,
\end{equation}
with $ \Xi^2_\pm(\mathcal{M}_4)$ space of self-dual and anti-self-dual harmonic $2$-forms respectively. We can denote the dimensions of these spaces with $b_+={\rm dim} (\Xi^2_+(\mathcal{M}_4))$ and $b_-={\rm dim} (\Xi^2_-(\mathcal{M}_4))$. The second Betti number $b_2$ of $\mathcal{M}_4$ is provided by the sum $b_2=b_+ + b_-$, while their difference,
\begin{equation}
\tau(\mathcal{M}_4)=b_+ \textcolor{blue}{-} b_-\,,
\end{equation}
provides the signature of $\mathcal{M}_4$.

Moving back to the consideration of the simplest Dirac operator in terms of $d+\delta$, and then to the definition of index of $\mathcal{M}_4$ in \eqref{index}, given the definitions we just summarized we can write:
\begin{equation}
{\rm Ker}((d+\delta)^+) = \Xi^0(\mathcal{M}_4) \oplus \Xi^2(\mathcal{M}_4) \oplus \Xi^4(\mathcal{M}_4) \,,
\quad
{\rm Ker}((d+\delta)^-) = \Xi^1(\mathcal{M}_4) \oplus \Xi^3(\mathcal{M}_4) \,,
\end{equation}
and finally:
\begin{equation}
{\rm Ind}((d+\delta)^+)= b_0+b_2+b_4 - (b_1+b_3)\,,
\end{equation}
which is the Euler characteristic of $\mathcal{M}_4$. Thus, for the simplest Dirac operator on $\mathcal{M}_4$, the index ${\rm Ind}$ of $\mathcal{M}_4$ equals the Euler characteristic of $\mathcal{M}_4$.

The Atiyah-Singer Index Theorem generalizes this result to the index of any first order elliptic linear operator, expressing it in terms of topological data. The Dirac operator $D_A$ with coefficients in a line bundle is of great importance for our discussion. For this latter one it is possible to write the decomposition:
\begin{equation}
D_A^+: \Gamma (W_+ \otimes L) \rightarrow \Gamma (W_- \otimes L)\,,
\quad
D_A^-: \Gamma (W_- \otimes L) \rightarrow \Gamma (W_+ \otimes L)\,,
\end{equation}
which are formal adjoint operators to one another. The kernels of $D_A^+$ and $D_A^-$ are finite-dimensional complex vector spaces. The index of $D_A^+$ is then defined to be:
\begin{equation} \label{index-Dirac}
{\rm Ind}(D_A^+)={\rm dim}({\rm Ker}(D_A^+)) - {\rm dim}({\rm Ker}(D_A^-))\,.
\end{equation}

The following theorems can be finally enunciated:

\begin{theorem} \label{atiyah-singer-1}
{\bf Atiyah-Singer Index Theorem (for the Dirac operator with coefficients in a line bundle).} Denoting with $D_A$ the Dirac operator with coefficients in a line bundle $L$ on a compact $4$-dimensional manifold $\mathcal{M}_4$, it holds that
\begin{equation}
{\rm Ind}(D_A^+)= -\frac{1}{8} \tau(\mathcal{M}_4) +\frac{1}{2} \int_{\mathcal{M}_4} c_1(L)^2\,,
\end{equation}
with $\tau(\mathcal{M}_4)=b_+- b_-$ signature of $\mathcal{M}_4$ and $c_1(L)$ the first Chern class of the component of the connection $A$ over $L$.
\end{theorem}

Here we have introduced the characteristic real valued $2k$-forms $\tau_k(A)$. These are globally defined and cast in terms of the field-strength $F$ of the connection $A$ over the bundle $W\otimes L$ as:
\begin{equation} \label{char}
\tau_k(A)= {\rm Tr}[(\frac{\imath}{2\pi} F )^k]\,,
\end{equation}
which are gauge-invariant expressions independent on the choice of the transition functions. In other words, for any $F_{\alpha}$ locally defined on $U_\alpha$ and transition functions $g_{\alpha \beta}$ over $U_\alpha \cap U_\beta$, given \eqref{char} it is easy to convince ourselves that:
\begin{equation} \label{char-g-i}
\tau_k(A)= {\rm Tr}[(\frac{\imath}{2\pi} F_\alpha )^k]={\rm Tr}[(\frac{\imath}{2\pi} g_{\alpha \beta} F_\alpha g_{\alpha \beta}^{-1} )^k]= {\rm Tr}[(\frac{\imath}{2\pi} F_\alpha )^k]= \tau_k(A) \,.
\end{equation}
The first and second Chern classes, respectively $c_1(A)$ and $c_2(A)$, are hence defined as:
\begin{equation} \label{chern-classes}
c_1(A)=\tau_1(A)
\qquad
c_2(A)=\frac{1}{2}[\tau_1(A)^2 - \tau_2(A) ] \,.
\end{equation}

The Atiyah-Singer Index Theorem (i.e. Thm. \ref{atiyah-singer-1}) with coefficients in a line bundle can be generalized in order to account for the index of Dirac operators of the type:
\begin{equation} \label{Dirac-E}
D_A^+: \Gamma(W_+ \otimes E) \rightarrow \Gamma(W_- \otimes E)  \,,
\end{equation}
with coefficients in a general complex vector bundle $E$. This generalization requires the introduction of the Chern character and of $\hat{A}$-polynomial in the Pontryagin classes.

\begin{theorem} \label{atiyah-singer-2}
{\bf Atiyah-Singer Index Theorem (over $4$-manifolds).} Given $D_A$ a Dirac operator with coefficients in a virtual vector bundle $E$ over a compact oriented $4$-dimensional manifold $\mathcal{M}_4$, then:
\begin{eqnarray}
\Index(D_A^+) &=& \hat{A}(TM)~ \ch(E)[\mathcal{M}_4] \nonumber \\
&= & \int_{\mathcal{M}_4} \left( -\frac{1}{24} (\dim E)\, p_1(TM) +\frac{1}{2} (c_1(E))^2 -c_2(E) \right)\,.
\end{eqnarray}
\end{theorem}

Further details on these constructions are provided in Appendix 3.\\

Having summarized the role of the Dirac operator in determining the topological data of a manifold, we can briefly comment how to extend this formalism to TQNNs. In order to establish a link between the tools provided by the Atiyah-Singer Index Theorems \ref{atiyah-singer-1} and \ref{atiyah-singer-2} and the determination of the topological data of the boundary manifolds $\mathscr{B}$, of which the $1$-complexes of TQNNs represent the dual one-skeleton graphs, we consider a result reported in \cite{morales-tecotl:94,morales-tecotl:95} within the loop quantization of the Einstein-Weyl theory of massless 2-components spinors coupled to gravity. At the Hamiltonian level, the theory can be summarized, respectively, by the Gau\ss  ~constraint $G_{AB}$, the vector diffeomorphism constraint $V_a$ and the time reparametrization (scalar) constraint $C$, namely:
\begin{eqnarray}
G_{AB} \!\!&=& \!\! -  \imath \sqrt{2} \, (D_A)_a \sigma^{a}_{\ AB} + \pi_{(A} \psi_{B)}\,, \\
V_a \!\!&=&\!\!  \imath \sqrt{2} \, \sigma^{b\,AB}\, \, F_{ab BA} - \pi_A  (D_A)_a \psi^A \,, \\
C=C_{\rm gr} + C_{\rm Weyl} \!\!&=&\!\! \sigma^{a\,AB}\, \sigma^{b\  \  C}_{\ B} \, F_{ab CA} + \imath \sqrt{2} \sigma^{a \ \ B}_{ \ A}  \pi_B\, (D_A)_a \psi^A
\,,
\end{eqnarray}
with $A_a^{AB}$ and $\sigma^a_{AB}$ canonically conjugated gravitational Ashtekar $\SU(2)$-symmetric variables, $A=1,2$ spinorial indices and $a=1,2,3$ space indices, and $\psi^A$ and $\pi^A$ respectively spinor field and its canonically conjugated variable. Considering open loop states, i.e. holonomies $H_{\alpha A}^{\ B}[A]$, with indices saturated by $\psi^A$ and $\pi^A$, namely:
\begin{equation}
X[A,\alpha] = \psi^A(\alpha_i) \, H_{\alpha A}^{\ B}[A] \, \psi_B(\alpha_f)\,,
\qquad
Y[A,\alpha] = \psi^A(\alpha_i)\,  H_{\alpha A}^{\ B}[A]\, \psi_B(\alpha_f)\,,
\end{equation}
the action of the gravitational scalar constraint $C_{\rm gr}$ on $X[A,\alpha]$ and $Y[A,\alpha]$ is equivalent to the action of the whole scalar constraint $C=C_{\rm gr} + C_{\rm Weyl}$ on $X[A,\alpha]$ and $Y[A,\alpha]$. In other words, only taking into account the action of the gravitational scalar constraint $C_{\rm gr}$, cast in terms of the densitized frame field $\sigma^a_{AB}$ and the field strength of the connection $A_a^{AB}$, generates the action of the $C_{\rm gr}$ and $C_{\rm Weyl}$, the latter one expressed in terms of the conjugated spinor variables and the Dirac operator $D_A$.

Given that quantum states on the boundary manifolds $\mathscr{B}$ can be expressed either within the spin-network basis or within the equivalent loop basis \cite{rovelli:95}, the Dirac operator $D_A$ is automatically introduced in the TQNN set up, provided that the dynamics is regulated by a $BF$-like extension of TQFTs to QFT  that reproduces the constraints of the Einstein-Weyl theory, along the lines of section \ref{gauge-inv}.

According to the procedure outlined in \cite{morales-tecotl:94,morales-tecotl:95}, the imposition on a loop state $\Phi_\alpha[A]$ of the gravitational scalar constraint $C_{\rm gr}$, which in terms of the smearing function $N$ casts:
\begin{eqnarray} \label{c-grav}
\widehat{C}_{\rm gr}[N] \Phi_\alpha[A] &=& \int d^3x N(x)~ {\rm Tr}\left( F_{ab} \frac{\delta}{\delta A_a}  \frac{\delta}{\delta A_b}  \right)\, {\rm Tr} [ P e^{\oint_\alpha A } ] \nonumber\\
&=& \int dt \int ds N(\alpha(s)) \delta^{3} (\alpha(s), \alpha(t)) \dot{\alpha}^a(s) \frac{\delta}{\delta \alpha^a(s)} \Phi_\alpha[A]\,,
\end{eqnarray}
turns out to be the same action of the total constraint $C=C_{\rm gr} + C_{\rm Weyl}$, encoding the Dirac operator $D_A$, on open holonomies with saturated ends:
\begin{eqnarray} \label{c-tot}
& \phantom{a}& \int dt \int ds N(\alpha(s)) \delta^{3} (\alpha(s), \alpha(t)) \dot{\alpha}^a(s) \frac{\delta}{\delta \alpha^a(s)}  \left( \psi^A(\alpha_i) \, H_{\alpha A}^{\ B} \, \psi_B(\alpha_f) \right) \nonumber\\
&=& \int d^3x N(x) \left[ {\rm Tr}\left( F_{ab} \frac{\delta}{\delta A_a}  \frac{\delta}{\delta A_b}  \right)
+ 2\, \dot{\alpha}^a (D_A)_a \psi_A(x) \frac{\delta}{\delta \psi_A(x)}
\right] \, \Phi_\alpha[A]  \nonumber\\
&=& \int d^3x N(x) \left[ {\rm Tr}\left( F_{ab} \frac{\delta}{\delta A_a}  \frac{\delta}{\delta A_b}  \right)
+ 2 \frac{\delta}{\delta A_a}  \frac{\delta}{\delta \psi_A(x)} (D_A)_a \psi_A(x) \right] \, \Phi_\alpha[A]  \nonumber\\
&=&  (\widehat{C}_{\rm grav}[N] + \widehat{C}_{\rm Weyl}[N]) \, \Phi_\alpha[A]= \widehat{C}[N] \, \Phi_\alpha[A] \,,
\end{eqnarray}
where in both \eqref{c-grav} and \eqref{c-tot} the canonically conjugated variables $\sigma^{a}_{AB}$ and $\pi^A$ have been implemented as left-invariant derivatives, respectively $ \frac{\delta}{\delta A_a^{AB}}$  and $\frac{\delta}{\delta \psi_A}$.

We may recast the result of Refs.~\cite{morales-tecotl:94,morales-tecotl:95} within the extended $BF$ formalism, which allows us to capture non-topological theories as well, taking into account a $\SL(2,\mathbb{C})$-bundle over a $4$-dimensional Lorentzian manifold $\mathcal{M}_4$. At this purpose, we observe that issues about the inclusion of spinor fields in the $BF$ formalism were previously reported and discussed in Refs.~\cite{jacobson:88,capovilla:91,alexander:14}, accounting for several different physical contexts, including the unification of forces and the explanation of the origin of the weak interaction's chirality \cite{alexander:14}. In the $BF$ formalism, the action for Weyl spinors $\psi_A$ can be written:
\begin{equation} \label{s-weyl}
\mathcal{S}_{\rm Weyl} = \int_{\mathcal{M}_4} \langle B \wedge \xi \wedge (D_A) \psi \rangle\,,
\end{equation}
in which the frame field is represented, in terms of the spinorial indices $A,A'=1,2$, as $B^{AB}=e^{A C'}\wedge e^B_{\ \ C'}$, and a spinorial one form $\xi^A$ is introduced. In order to exclude that $\xi^A$ has spin $\frac{3}{2}$, simplicity constraints as in \eqref{simpli} of section \ref{gauge-inv} must be considered, together with an additional constraint term in the action \cite{capovilla:91}, namely:
\begin{equation}
\mathcal{S}_{\frac{1}{2}} = \int_{\mathcal{M}_4} \left( B^{AB} \wedge \xi_A \wedge D \psi_B  + \tau_{ABC} \wedge B^{AB} \wedge \xi^C \right) \,,
\end{equation}
where we have denoted now the Dirac operator as $D$, and with $\tau_{ABC}=\tau_{(ABC)}$ a multiplet of Lagrangian multipliers and a $1$-form symmetric in the spinor indices. Variation with respect to $\tau_{ABC}$ provides the constraint:
\begin{equation}
B^{(AB}\wedge \xi^{C)} =0
\,,
\end{equation}
which combined with the constraint \eqref{simpli} finally enforces the condition:
\begin{equation}
\xi^A=e^{AA'} \pi_{A'}
\,.
\end{equation}

In a background independent formalism, the $BF$ action complemented with \eqref{s-weyl} reads:
\begin{equation}
\mathcal{S}= \mathcal{S}_{\rm BF} + \mathcal{S}_{\rm Weyl}
= \int_{\mathcal{M}_4} \langle B \wedge F + B \wedge \xi \wedge D_A \psi \rangle \,.
\end{equation}

Likewise to that to be specified in section \ref{gauge-inv}, we arrive at the Einstein-Weyl theory, complemented with the Holst topological term for the action of gravity, introducing a multiplet of Lagrangian multipliers $\phi_{AA'\ BB' \ CC' \ DD' }$.

The inclusion of fermions in the TQNN sections of the TQN2C is consistent with holographic dualities between SYK theories on the boundary manifolds $\mathscr{B}$ and the Schwarzian mechanics on one side, and TQFT on the bulk manifolds $\mathcal{M}$, which has been explored with significant depth in the literature \cite{gross:17,stanford:17,mertens:17,turiaci:17,kitaev:18,mertens:18,blommaert:18,gaikwad:20}. Thus it is not surprising, as we will comment in the next section, that both the SYK model and TQNNs can be applied to the specific context of condensed matter physics and information theory \cite{sachdev:93,kitaev:15a,kitaev:15b}.

\section{Applications}\label{app}

Sequential measurements provide the empirical foundation not just of physics, but of all sciences.  As shown above, sequential measurements, without further assumptions, induce TQFTs. Thus TQFTs provide an underlying ``default'' theoretical structure for physics, computer science, and other scientific disciplines. This structure is, significantly, free of background assumptions about either spacetime or the ``ontology'' of systems of interest.  Measurements of position and momentum, in particular, become measurements of a ``spacetime'' sector of the observational Hilbert space $\mathcal{H}_A$ of some observer $A$.  Spacetime is, therefore, in this framework a quantum system by definition, and is observer-relative in the ``personal'' sense advocated by QBists \cite{fuchs:03, fuchs:13, mermin:18}.  This measurement-centric perspective casts a number of well-studied phenomena in a new light; we consider some of these below, deferring others to future work.

\subsection{Gauge invariance and effective field theories} \label{gauge-inv}

It has previously been shown \cite[Thm. 2]{addazi:21} that in any system $U = AB$ for which the interaction $H_{AB}$ can be written as Eq. \eqref{ham}, the bulk interactions $H_A$ and $H_B$ respect gauge invariance. This result is clearly applicable in the present setting, where it can be reformulated as:

\begin{theorem} \label{gauge-thm}
Any effective field theory (EFT) $\mathscr{T}$ on a sector $S$ that is consistent with the TQFT induced by sequential measurements of $S$ is gauge invariant.
\end{theorem}

\begin{proof}
An EFT on $S$ is a map $\mathscr{T}: S_1 \rightarrow S_2$, where $S_1$ is the spacetime sector of $S$ and $S_2$ is some other sector of $S$, the outcome values from which are treated as dependent on spacetime.  We can equally well write $\mathscr{T}: \mathcal{H}_{S_1} \rightarrow \mathcal{H}_{S_2}$.  Hence $\mathscr{T}$ can fail to be gauge invariant only if the joint-system evolution operator $\mathcal{P}_U$ fails to be gauge invariant.  This, however, is ruled out by \cite[Thm. 2]{addazi:21}.
\end{proof}

The notion of EFT we have introduced captures any quantum field theory (QFT) applied either to the description of fundamental forces within the standard model of particle physics or to low-energy scale phenomena in state solid physics. Generally speaking, an EFT describes the propagation of either dynamical or collective degrees of freedom. The addition of either dynamical or collective degrees of freedom can be achieved within the framework of TQFTs by imposing extra constraints to the phase space variables of the system. For the sake of completeness in the exposition, we briefly review how to cast classically, within the $BF$ formalism, Yang-Mills theories and the Einstein-Hilbert action for gravity in dimensions higher than three.  In three dimensions, the Einstein-Hilbert action of gravity does not posses dynamical degrees of freedom and directly casts as a $BF$ theory.

In $d$-dimensions (with $d>2$), the density Lagrangian of pure Yang-Mills theories, for any gauge group $G$, is expressed in the first order formalism by \cite{capovilla:91,fucito:97}:
\begin{equation}\label{BF-YM}
\mathcal{L}_{\rm YM}[A,B]=\langle B\wedge F +\frac{g_{\rm YM}^2}{4} B\wedge \star B\rangle,
\end{equation}
where $g_{\rm YM}$ is the bare coupling constant of the Yang-Mills theory under scrutiny and the Hodge dual $\star$ of a $p$-form in $d$-dimensions is defined as $\star=\varepsilon^{i_1 \dots i_p}/(d-p)!$, with $\varepsilon^{i_1 \dots i_p}$ the (densitized) Levi-Civita symbols of the spacetime indices. The equations of motion for the theory defined by \eqref{BF-YM} read:
\begin{equation}\label{EOM-BF-YM}
F=\frac{g_{\rm YM}^2}{2} \star B\,, \qquad d_AB=0\,.
\end{equation}
Using the curvature constraint, or in other words performing the path integral to integrate out the frame field $B$, i.e. calculating the partition function $\mathcal{Z}$ of the theory: 
\begin{equation}
\mathcal{Z}_{\rm YM}[A,B]= \int \mathcal{D}A\, \mathcal{D}B \exp ( \imath \mathcal{L}_{\rm YM}[A,B])
=  \int \mathcal{D}A\, \exp ( \imath \mathcal{L}_{\rm YM}[A])
\,,
\end{equation}
yields the customary action of a Yang-Mills theory in the first order formalism, namely:
\begin{equation}\label{EOM-BF-YM2}
\mathcal{L}_{\rm YM}[A]=\frac{1}{g_{\rm YM}^2} \langle F \wedge \star F \rangle \,.
\end{equation}
This construction clearly encodes any $SU(N)$ Yang-Mills theories. It also captures abelian $U(1)$ gauge theories, and hence electro-magnetism. Thus the whole gauge sector of the standard model can be cast within the $BF$ formalism, which provides a framework for EFTs.

The Einstein-Hilbert theory of gravity can be formulated in the $BF$ formalism as well. In $3$-dimensions, the action of gravity is a $BF$ theory of the form \eqref{BF}, in which the $B$ field is the frame field $e$, the drei-bein in this case, that individuates the ``square root'' of the gravitational field, and the mixed-indices Riemann tensor $R$ is the field strength $F$ of the spin connection $\omega=A$, namely:
\begin{equation}
\mathcal{S}^{3d}_{\rm EH}[e,\omega]=\int_{\mathcal{M}_3} \langle e \wedge R[\omega] \rangle\,,
\end{equation}
where the internal gauge group is either $\SO(2,1)$ or $\SU(1,1)$ for Lorentzian manifolds, and either $\SO(3)$ or $\SU(2)$ for Euclidean manifolds.

In four dimensions, the expression of the Einstein-Hilbert action of gravity in the $BF$ formalism is more subtle, as it requires to ``unfreeze'' the two dynamical degrees of freedom of the graviton. This is attained by means of the imposition of the simplicity constraint, which requires the $2$-form field $B^{IJ}$ --- labelled with an anti-symmetrized pair of indices $IJ$, with $I,J=1,\dots4$, that is in the adjoint representation of $\SO(3,1)$ (in the Lorentzian) or $\SO(4)$ (in the Euclidean) --- to be a bi-vector. Introducing a Lagrangian multiplier $\phi_{IJKL}$, which is symmetric under the exchange of the pairs of anti-symmetric indices $IJ$ and $KL$, one can then write the Plebanski action \cite{plebanski:77} for gravity in $4$-dimensions:
\begin{equation}\label{Pleb}
\mathcal{S}_{\rm Pleb}[B,A,\phi]= \int_{\mathcal{M}_4}  B^{IJ}\wedge F^{IJ} +\phi_{IJKL} \, B^{IJ} \wedge B^{KL}  \,.
\end{equation}
Variation of \eqref{Pleb} with respect to $\phi_{IJKL}$ provides the simplicity constraint:
\begin{equation}\label{simpli}
 B^{(IJ} \wedge B^{KL)}=0\,,
\end{equation}
which has as solutions the reduction of the $B$ field to be a bi-vector of either the form $e^I\wedge e^J$ or $\epsilon^{IJ}_{\ \ KL} e^K\wedge e^L$, with $\epsilon^{IJ}_{\ KL}$ Levi-Civita symbols for the internal indices. Denoting the spin connection $A^{IJ}$ with $\omega^{IJ}$ and recognizing that its field strength $F^{IJ}[A]$ provides the mixed indices Riemann tensor $2$-form $R^{IJ}[\omega]$, the former solution are easily seen to introduce the topological Holst action, one of the topological sectors for gravity, while the latter solution reproduces the Einstein-Hilbert action. Summing the two aforementioned contributions, we recover the Einstein-Hilbert-Holst action:
\begin{equation}\label{BF-EH-4d}
\mathcal{S}_{EHH}[e,\omega]= \frac{1}{\kappa^2} \int_{\mathcal{M}_4}  \epsilon_{IJKL}\, e^I \wedge e^J \wedge R^{KL}[\omega] + \frac{1}{\gamma} \int_{\mathcal{M}_4}
e_I \wedge e_J \wedge R^{IJ}[\omega]  \,,
\end{equation}
with $\gamma$ the Barbero-Immirzi parameter and
$\kappa^2=16 \pi G_N$, having denoted with $G_N$ the Newton constant. 

Notice that the quantization method we have summarized in \S \ref{CCCD-TQNN} for $BF$ theories has been extended to quantum gravity in the ``loop'' approach \cite{rovelli:04}, by casting the theory over the $\SU(2)$-bundle individuated by the use of the Ashtekar variables \cite{ashtekar:86, ashtekar:87} once the time gauge is fixed.



\subsection{Contextuality as a resource}

Quantum contextuality is widely recognized as a resource for quantum computation \cite{howard:14, bermejo:17, frembs:18}.  Sequential measurements of the form of Eq. \eqref{flow-2} and their induced TQFTs of the form of Diagram \eqref{PQ-TQFT} represent context switches in which mutually-noncommuting sets of measurement operators are deployed on a single sector.  Contextuality in this case is induced by noncommutativity of the CCCDs representing non co-deplyable sets of operators, as made precise below.

\subsubsection{Quiver representations and their sections}\label{quiv-sect}

First of all, we recall the basic definitions of a quiver and its representation \cite{derksen:05} as applied in Ref.~\cite{seigal:21}. A \emph{quiver} $Q$ is a directed graph consisting of a pair $(V,E)$, where $V$ denotes a finite set of vertices and $E$ a finite set of edges (arrows), and two maps $s,t: E \lra V$, the source (tail) and target (head), respectively. For a fixed $Q$ and some base field, we tether to each $v \in V$ a finite-dimensional vector space $\mbA_v$, and a linear map to each edge (with suitable domain and codomain). Specifically, a \emph{representation} $\mbA_{\bullet}$ of $Q$ consists of a collection $\{\mbA_v: v \in V\}$ of finite-dimensional vector spaces, and a collection of linear maps $\{\mbA_e: \mbA_{s(e)} \lra \mbA_{t(e)}: e \in E \}$.

As in Ref.~\cite{seigal:21}, on taking a path $p = (e_1, e_2, \ldots, e_k)$ in $Q$ of distinct edges and sources, we associate to each such path $p$ the map $\mbA_p: \mbA_{s(p)} \lra \mbA_{t(p)}$, defined via
\begin{equation}
\mbA_p := \mbA_{e_{k}} \circ \mbA_{e_{k-1}}  \circ \cdots \circ \mbA_{e_{2}} \circ \mbA_{e_{1}} \,.
\end{equation}
The total space of $\mbA_{\bullet}$ is the direct product $\text{Tot}(\mbA_{\bullet}) : = \prod_{v \in V} \mbA_v$.

We next turn to the definition of a \emph{section} as given in \cite{seigal:21}:
\begin{definition}\label{section-def}
Let $\mbA_{\bullet}$ be a representation of $Q = (s,t: E \lra V)$. A section of $\mbA_{\bullet}$ is an element $\gamma = \{\gamma_v \in \mbA_v: v \in V\}$ in  $\text{Tot}(\mbA_{\bullet})$ satisfying the compatibility requirement $\gamma_{t(e)} = \mbA_e(\gamma_{s(e)})$ across each edge $e$ in $E$.
\end{definition}
The set of all sections of $\mbA_{\bullet}$ is a vector subspace of $\text{Tot}(\mbA_{\bullet})$ denoted $\Gamma(Q, \mbA_{\bullet})$. As noted in \cite{seigal:21}, Definition \ref{section-def} commences an analogy between quiver representations and the theory of sheaves and vector bundles as they arise in algebraic geometry.

\subsubsection{Measurement and noncommutativity}\label{measure-1}

In putting this brief account of quiver representations to work, we view a CCCD in Diagram \eqref{cccd-1} as a digraph, following Definition \ref{def-cat-CCCD}, and identify it with its underlying quiver $Q= (V,E)$ as above in a ``basic'' form: specifically, vertices $V$ consist of classifiers, and edges $E$ are infomorphisms between classifiers. Here we do not consider ``self infomorphisms'' implemented by loops at CCCD vertices; hence the quiver description of the CCCD can be considered ``basic'' as depicted in Eq. \eqref{cccd-1} and \eqref{cccd-2}. Next, we can reasonably assume that measurements are modeled by a probability space $\mfP = (\mfO, \mfM, \mfp)$, where $\mfO$ denotes a set of outcomes, $\mfM$ a $\sigma$-algebra of events, and $\mfp$ a probability measure. Then consider the vector space of measurable functions (random variables) $\mfp(f): \mfP \lra \mathbb{R}$ \cite{small:94}. These vector spaces will be the $\mbA_v$ tethered to each vertex $v \in V$ in the CCCD.

We now recall the sequence of measurements in Eq.~\eqref{flow-1} and \eqref{flow-2} together with the collection of maps $\{\mfF(i)\}$ leading to the cobordisms in Diagram \eqref{CCCD-to-Cob} and \eqref{CCCD-to-Cob-2}, respectively. Suppose at some level $k$ in a sequence of measurements $S_{k} \lra S_{k+1}$ is not defined, either by disturbance, ambiguity, or is possibly non-empirical (i.e. fails to commute as in Eq. \eqref{not-CCCD}). Then the map $\mfF(k)$ is not defined, and this leads to either a `singularity' in the cobordism, or the latter is not definable.  In this case Diagrams \eqref{CCCD-to-Cob} and \eqref{CCCD-to-Cob-2} fail to commute. In particular, the associated CCCD fails to commute. Prescribing sections as in Definition \ref{section-def}, the situation can be summarized as follows.
\begin{theorem}\label{section-th}
If the CCCD in Diagram \eqref{cccd-1}, when viewed as a quiver, is noncommutative (and hence not a CCCD), then at least one section of a quiver representation by vector spaces is not definable; in particular, at least one section of a representation by vector spaces of measurable functions is not definable.
\end{theorem}
\begin{proof}
Following the proof of \cite[Th. 7.1]{fg:22}, if the CCCD fails to commute then there is at least one vertex, such as the core $\mbC$ (colimit) classifier, and an edge from it, that are not defined. Hence, the compatibility condition imposed for a section in Definition \ref{section-def}
does not hold in that case.
\end{proof}
This is in essence saying that in the presence of noncommutativity, $\Gamma(CCCD, \mbA_{\bullet})$ admits no ``global'' section of the representation definable across the entire CCCD.
The noncommutativity derived from \cite[Th. 7.1]{fg:22} amounts to an intrinsic (quantum) contextuality imposed on a corresponding set of observables
considered as non-co-deployable.
We note that Theorem \ref{section-th} bears some similarity to the main result of \cite{abramsky:11} proving that the non-existence of a global section of a sheaf of measurable distributions implies Kochen-Specker contextuality.

\subsubsection{Quantitating contextuality by QRF dimension reduction}

The above construction suggests a simple measure of relative contextuality for pairs $(\mathbf{X}, \mathbf{Y})$ of QRFs.  Letting $\mathbf{X}$ and $\mathbf{Y}$ also name the respective CCCDs, hold $\mathbf{X}$ fixed and consider maps $\xi_i: \mathbf{Y} \mapsto \mathbf{Y}_i$, where $\mathbf{Y}_i$ is a CCCD that embeds in $\mathbf{Y}$.  The action of each $\xi_i$ is to remove one or more base-level classifiers, and the associated maps, from $\mathbf{Y}$ to produce $\mathbf{Y}_i$.  Considering all such maps, we can define:

\begin{definition} \label{context-dimension}
The {\em contextuality dimension} $\dim(\mathbf{Y} \vert \mathbf{X})$ of $\mathbf{Y}$ relative to $\mathbf{X}$ is $\text{min}(\dim(\mathbf{Y}) - \dim(\mathbf{Y}_i))$ such that $\mathbf{X}, \mathbf{Y}_i \mapsto \mathbf{XY}_i$ is a CCCD morphism, where the minimum is computed over all maps $\xi_i: \mathbf{Y} \mapsto \mathbf{Y}_i$ such that $\mathbf{Y}_i$ embeds in $\mathbf{Y}$.
\end{definition}
\noindent
If $\mathbf{X}$ and $\mathbf{Y}$ are co-deployable, $\dim(\mathbf{Y} \vert \mathbf{X})$ is clearly zero: no classifiers need to be removed from $\mathbf{Y}$ to construct a CCCD $\mathbf{XY}$.  If $\mathbf{X}$ and $\mathbf{Y}$ are not co-deployable, $\dim(\mathbf{Y} \vert \mathbf{X}) \geq 1$.

It is interesting to compare this with the approach of \cite{abramsky:17} where a ``measurement scenario'' is formalized as an empirical model $e$ that is specified by a probability distribution expressed as a convex combination of a non-contextual model $e^{\text{NC}}$ and a ``no-signalling'' model $e'$, i.e. $e= \lambda e^{\text{NC}} + (1 - \lambda) e'$, with $\lambda \in [0,1]$.
In relationship to a global probability distribution on the outcomes of all measurements, the maximum possible, admissible value of $\lambda$ in such a decomposition of the local model $e$ is called the \emph{non-contextual fraction} $\mathsf{N} \mathsf{C} \mathsf{F}(e)$ of $e$; the \emph{contextual fraction} ($\mathsf{C} \mathsf{F}$) is then $\mathsf{C} \mathsf{F}(e) := 1 - \mathsf{N}\mathsf{C} \mathsf{F}(e)$. Thus, the empirical model $e$ is said to be \emph{contextual} if the corresponding family of probability distributions in \cite{abramsky:17} cannot itself be obtained as the marginals of the global probability distribution on the outcomes to all measurements (cf. \cite{dzha:17b,gudder:19}). In view of Definition \ref{context-dimension}, the present development of ideas suggests a discrete version of the contextual fraction in \cite{abramsky:17}. Here $1-\mathsf{C} \mathsf{F}$ would correspond to the outcomes from $\mathbf{XY}_i$.  $\mathsf{C} \mathsf{F}$ is then the remainder, the outcomes from the classifiers that have to be removed.  It also suggests that in a large-dimension limit, $\dim(\mathbf{X}) +\dim(\mathbf{Y})$ can be normalized towards noncontextuality, consistent with effective classicality in this limit.

\subsubsection{Quiver representations, gauge-networks and noncommutative geometry} \label{quivers-gauge-networks}

The quiver representation enables to connect the framework we developed in the previous sections to noncommutative geometry. This passes through the introduction of a generalized version of spin-network that goes under the name of gauge-networks \cite{marcolli:14}, which intertwine among the notions of spin-network in quantum gravity and quantum neural network, lattice gauge-theory, and almost-commutative geometries, used for the description of particle physics models in noncommutative geometry.

To approach matter-gravity interaction, in noncommutative geometry one considers locally the product among an ordinary $4$-dimensional spacetime manifold and a finite spectral triple, the latter being the generalization to noncommutative geometry of a compact spin manifold. A spectral triple is defined by the data $(\mathcal{A},\mathcal{H},D)$, which comprise the involutive algebra $\mathcal{A}$, encoding bounded operators represented on a Hilbert space $\mathcal{H}$, and a Dirac operator $D$, fulfilling the compatibility condition that commutators with elements in the algebra are bounded.

When both $\mathcal{A}$ and $\mathcal{H}$ are finite dimensional, one recovers the case of metrically zero dimensional noncommutative spaces. Almost-commutative geometries can be then introduced as product spaces of a finite spectral triple and an ordinary manifold, which can also be seen as a spectral triple. The spectral action, which is a natural action functional on these spaces, enables in the asymptotic expansion to recover the action for gravity coupled to matter. In turn the Lagrangian for matter can be determined by selecting certain finite noncommutative spaces --- see e.g. Refs.~\cite{chamseddine:96,chamseddine:07,chamseddine:10}.

As proposed in Ref.~\cite{marcolli:14}, a procedure similar to the one outlined in \S \ref{g-bundle-metric}, while discretizing manifolds and then assigning geometrical data to $1$-complexes, can be extended to almost-commutative geometries, by assigning finite spectral triples to the vertices of graphs and morphisms to their edges/links. This leads to introduce gauge-networks, which can be thought as quanta of noncommutative space, exactly as spin-network have been thought a quanta of $3$-dimensional space manifolds in the loop approach to quantum gravity \cite{rovelli:04}. It has been delved in the literature for some decades that the Yang-Mills action can be recovered in the continuum limit from the Wilson action cast in terms of holonomies \cite{creutz:83}. It can be further demonstrated --- see e.g.  Ref.~\cite{marcolli:14} --- that the Wilson action can be recovered from the spectral action of the Dirac operator acting on gauge-networks, and that this procedure enables to account for additional terms in the action that provide a lattice gauge theory coupled to the Higgs field in the adjoint representation \cite{lang:81,drouffe:84}.  Some further details are provided in Appendix 4.

\subsection{Application to condensed matter physics}
\noindent
We provide an overview of potential phenomenological applications of the TQFT framework, and in general of quiver representations, recently pointed out within several contexts of condensed matter physics. We focus on few notable examples, provided by topological insulators, fracton phases of matter and trapped atoms and ions. The extended framework we have provided here suggests that a novel link between sequential measurements of condensed systems within a QRF formalism and TQFTs can be pushed forward in the future.

\subsubsection{Topological insulators, string-nets and $BF$ theories}
\noindent
Topological phases of electrons with time-reversal symmetry have been studied over more than a decade. These systems have been called topological insulators \cite{hasan:10,moore:10}, and are induced by the spin-orbit coupling, appearing both in $2$-dimensional \cite{kane:05,bernevig:06,koenig:07} and in $3$-dimensional space \cite{fu:07,moore:07,roy:09,hsieh:08}. Topological insulators are characterized by the properties of being insulating in the bulk, while conversely supporting conducting edge or surface states.

Topological phases of matter have been described resorting to TQFTs, in a way that is reminiscent of the understanding of the symmetry-breaking phases through the Landau-Ginzburg mechanism. In a similar way, topological insulators in two and three space-dimensions can be described by TQFTs, in particular by $BF$ theories, allowing to extend the link to QRFs and sequential measurements.

The electromagnetic response, in presence of weak time-reversal-breaking perturbations which gap the surface states the of $3$-dimensional topological insulators, has been studied in \cite{qi:08,essin:09}. This has subsequently allowed to find an effective field theory description for topological insulator phases in \cite{cho:11}, extending the Chern-Simons effective theory already adopted to explain the quantum Hall effect. This Chern-Simons theory for an abelian field $A$ is expressed by the Lagrangian:
\begin{eqnarray}
\mathcal{L}=\frac{k}{4 \pi} \epsilon^{\mu \nu \rho} A_\mu \partial_\nu A_\rho\,,
\end{eqnarray}
where $k$ is an integer number denoting the Chern-Simons Lagrangian density.

For topological insulators in $2$-dimensional space, we can consider a $BF$ theory over a $(2+1)$-dimensional spacetime manifold \cite{cho:11}, which notoriously recasts in terms of pairs of Chern-Simons actions. We have shown previously that a non abelian $BF$ theory over a $\SU(2)$-bundle is equivalent to the Einstein-Hilbert action for gravity in $(2+1)$-dimensions. This latter theory is indeed topological, even in presence of a cosmological constant, which enables to recast the action in terms of the difference of two Chern-Simons actions. When an abelian $BF$ is considered in $(2+1)$-dimensions, this can describe a fractional quantum Hall state, a phase that can consistently be realized from the pairing of integer quantum Hall states. At the observational level, these models also allow the description of both gapless surfaces, with time-reversal symmetry preserved, and situations for which the time-reversal symmetry is broken and surfaces are gapped \cite{cho:11}. Furthermore, $(3+1)$-dimensional $BF$ theories can be invoked as a description of topological insulators that are $3$-dimensional in space and are endowed with fractional statistics of point-like and line-like objects \cite{cho:11}.

The investigation of topological phases of matter in $(2+1)$-dimension has triggered the introduction of microscopic spin models called Levin-Wen models \cite{wen:05}. These are based on string-nets, namely colored graphs, and consist of a family of rigorously solvable lattice spin Hamiltonians that is applicable to a large part of topological phases in $2$-dimensional space. When the colouring is with respect to a unitary fusion category $\mathcal{C}$, topological phases turn out to be described by TQFTs that are based on the Drinfeld center $Z(\mathcal{C})$ .

The extension to three and higher space dimensions has been
also considered in \cite{wen:05}, using unitary symmetric fusion categories. A further generalization of the Levin-Wen models in $(3+1)$-dimensions that is based on unitary braided fusion categories has been discussed in \cite{walker:12}, in light of the possibility to make contact with fractional topological insulators and projective ribbon permutation statistics. Unitary braided fusion categories are non-trivial product of discrete gauge theories endowed with a unitary modular category. Deploying unitary modular categories instantiates a generalization of the $(3+1)$-dimensional $BF$ theories with cosmological-like term in the action.

Constructing a TQFT in $4$-dimensions is quite an arduous task. In $4$-dimensions there are few existing examples of smooth topological invariants, either classical or quantum, among which the TQFT constructed by Witten in \cite{witten:88}, which is a $\mathcal{N} = 2$ supersymmetric Yang-Mills theory reproducing the Donaldson invariant of $4$-manifolds. There exist also, as notable families of $(3+1)$-dimensional TQFTs,  finite gauge group theories \cite{dijkgraaf:90} and $BF$ theories \cite{baez:96}, both giving rise to topological invariants that are determined by classical homotopy invariants, and TQFTs based on unitary braided fusion categories \cite{walker:12}.

TQFTs in $(3+1)$-dimensions that are topological gauge theories constructed over a compact $G$-bundle, with $G$ compact, can be obtained for instance considering the Pontryagin density over $\mathcal{M}_4$:
\begin{eqnarray}
\mathcal{S}[A]= \frac{k}{8 \pi^2} \int_{\mathcal{M}_4} \langle F \wedge F  \rangle\,,
\end{eqnarray}
which recasts as the Chern-Simons action over the boundary $\mathcal{B}_3$ of $\mathcal{M}_4$, namely:
\begin{eqnarray} \label{WCS}
\mathcal{S}[A]= \frac{k}{8 \pi^2} \int_{\mathcal{B}_3} \langle A \wedge dA +\frac{2}{3} A\wedge A\wedge A  \rangle\,,
\end{eqnarray}
often referred to as Witten-Chern-Simons theory.
The cohomology groups and quantization of topological gauge theories with finite gauge groups have been studied in \cite{dijkgraaf:90}, where a link with $2$-dimensional holomorphic orbifold models has been established. In \cite{dijkgraaf:90} a natural map among $3$-dimensional Chern-Simons gauge theories and $2$-dimensional sigma models has been recovered and then generalized to topological spin theories that are defined on $3$-dimensional manifolds with spin structure and related to $2$-dimensional chiral superalgebras.

$BF$ theories in $(3+1)$-dimensions over a $G$-bundle and with cosmological term are expressed by the action:
\begin{eqnarray} \label{TQFT-cosmo}
\mathcal{S}_{\rm BF}= \int_{\mathcal{M}_4} \langle B \wedge F + \Lambda B \wedge B \rangle\,,
\end{eqnarray}
where $\Lambda$ plays a similar role to the cosmological constant when $G=GL(4,\mathbb{C})$. The partition function of this theory turns out to be:
\begin{eqnarray}
\mathcal{Z}_{\rm BF}(\mathcal{M}_4) = \int \mathcal{D}B \mathcal{D}A\,  e^{\imath \mathcal{S}_{\rm BF}} = e^{\imath \beta \sigma(\mathcal{M}_4)}\,,
\end{eqnarray}
with $\beta$ constant and $\sigma(\mathcal{M}_4)$ signature of $\mathcal{M}_4$, previously introduced.

TQFTs in $(3+1)$-dimensions with a cosmological-like term as in \eqref{TQFT-cosmo} are called Crane-Yetter theories \cite{crane:93a,crane:93b}, and can be thought as an `up-graded' version of $\SU(2)$-symmetric Witten-Chern-Simons theories. On the other hand, Crane-Yetter TQFTs, the braid group of which can be extended to account for more general mapping class groups, induce representations of the motion groups.

Modular categories constructed from Chern-Simons theories entails to recover Pontryagin densities $F\wedge F$ in the bulk, while Chern-Simons theories survive on the boundary. The relation in $(3+1)$-dimensions between $BF$ theories and TQFTs based on unitary modular categories can be understood by inspecting the illustrative case of $BF$ theory over the $\SU(2)$-bundle. Setting $\Lambda=\lambda/12$, the partition function recasts as:
\begin{eqnarray}
\mathcal{Z}_{\rm BF}(\mathcal{M}_4) = \int \mathcal{D}B \mathcal{D}A\,  e^{\imath \int_{\mathcal{M}_4} \langle B \wedge F +\frac{ \lambda}{12} B \wedge B \rangle }   = \frac{2\pi}{\sqrt{\lambda}}  \int  \mathcal{D}A\,  e^{\imath \int_{\mathcal{M}_4} \langle F \wedge F  \rangle } \,,
\end{eqnarray}
which compared with the $SU(2)$-symmetric Witten-Chern-Simons theory provides the identification of the Chern-Simons level with the inverse of the cosmological constant, i.e. $k=12 \pi/\lambda$. The limit for which the semi-classical theory is recovered is attained for $\lambda\rightarrow 0$, entailing $k \rightarrow \infty$.

Excitations that include points and defect loops can be considered, together with their statistics. When unitary modular categories are involved, being these algebraic theories of anyons, topological gauge theories turn out to be coupled to anyons \cite{chen:17,walker:12}. Several models for fractional topological insulators based on TQFTs have been proposed in the literature --- see also \cite{qi:10, swingle:10}. Demanding a breakdown of the time-reversal symmetry or the closing of the gaps, deformations to discrete gauge theories can be achieved for the fractional topological insulators, which are gapped phases of matter for fermions. It has been proposed in \cite{walker:12} that the underlying topological orders in the fractional topological insulators after symmetry-breaking corresponds to $(3+1)$-dimensional TQFTs based on a unitary braided fusion category.

\subsubsection{Quiver representations for fractons}
\noindent
Structural analogies can be also pointed out among quiver gauge theories and some lattice models, which enables to investigate fracton phases of matter. Here the starting point for the analysis is the observation that lattice models can be cast in the quiver formalism, incorporating symmetries of subsystem in the quiver structure. As an effective description of phases of condensed matter physics, lattice models can be introduced. These are either mechanical systems, the degrees of freedom of which are lattice points, edges and faces, realizing an example of $0+1$-dimensional QFT, or spatial arrangements of atom degrees of freedom, in which excitations propagate in real space. Within this second perspective, one may deploy quivers endowed with super-symmetric structure, and feature a linear increase of the moduli spaces with the size of the lattice, reproducing the excitations with limited mobility \cite{razamat:21}.

Supersymmetric quiver gauge theories are QFTs in $(d-1)+1$-dimensions, within which gauge fields are associated to vertices and edges are associated to matter fields. Lattice in the quiver theories can be regarded a the discretization of the spatial manifold. Thus the lattice provides within this perspective the $1$-complex topological structure of the quiver. Quivers hence represent stacks of $d$-dimensional coupled layers that have the dimension of the lattice.

There exist lattices that originate fractonic phases of matter \cite{vijay:15,nandkishore:20}, while retaining interesting properties that include subsystem symmetries, excitations of restricted mobility and a large number of vacua, the logarithm of which scales linearly with the size of the lattice. Lattices exhibiting fractonic phases are also relevant in shedding light on the application of the Wilsonian renomalization group flow from lattice UV theories to the continuum structures in the IR \cite{seiberg:21}. Further details are provided in Appendix 5.

\subsubsection{TQFTs and resonant ions}
\noindent
Engineering atom-laser interactions in linear crystals of trapped ions provide potential quantum simulations of complex physical systems, which in turn are described by effective model connected to TQFTs. Models that provide analog simulation of simple lattice gauge theories have been constructed, with a dynamics that can be mapped onto spin-spin interactions in any dimension \cite{davoudi:19}. In particular, possible instantiations include $(1+1)$-dimensional quantum electrodynamics, $(2+1)$-dimensional abelian Chern-Simons theories coupled to fermions, and $(2+1)$-dimensional pure $Z_2$ gauge theories. Scalable analog quantum simulation of Heisenberg spin models can be hence achieved in any number of space dimensions, while reproducing arbitrary interaction strengths.

At the same time, modifying the representation of gauge fields on qubit degrees of freedom has provided an interesting pathway with profound implications for $\SU(3)$ Yang-Mills gauge theory on a lattice of irreducible representations. This is a strategy that is often followed experimentally to access gauge-variant states, which solely retain a physical meaning. Within this framework, in \cite{ciavarella:21} both irreducible representations in a basis of projected global quantum numbers and in a local basis with controlled-plaquette operators supporting efficient time evolution have been considered.

\section{Conclusions and outlooks}\label{con}

We have shown here that sequential measurements of any finite system $S$ defined as the domain of one or more QRFs induce a TQFT on $S$.  Hence any measurable system can be given a background-free physical theory; one that makes no ``ontic'' assumptions about the system beyond measurability by some finite apparatus.

We have demonstrated, moreover, a functional relationship between two heretofore disparate formalisms: that of networks (CCCDs) of Barwise-Seligman classifiers, and that of finite cobordisms. While cobordisms are familiar in physics, such classifier networks have primarily been applied in natural-language semantics (the original application of \cite{barwise:97}), computational semantics and ontologies \cite{kalfoglou:03,schorlemmer:02,schorlemmer:05,fg:19b,dutta:19}, and the context-dependent theory of inference \cite{fg:22,ffgl:22} --- for a range of other examples cast in the isomorphic category of Chu spaces, see e.g. Ref.~\cite{fg:19a}. Hence our results lend credence to Grinbaum's \cite{grinbaum:17} claim that ``physics is about language'', or the QBist position that physics is about inference \cite{fuchs:13,mermin:18}, or more generally computation \cite{fuchs:03}.

Echoing what was said in the Introduction, the graph network (GN) development of \S\ref{CCCD-cat} leading to defining the category $\mathbf{CCCD}$, affords further applications to ANNs and VAEs (cf. \cite{armenta:21} in terms of quiver representations). Such GNs are embraced by the general theory of cell complexes \cite{hatcher:05} to which these apply (see e.g. \cite{hajij:20}), and our representation of sequential measurements and TQFTs opens the door for future work in this and other directions.

The extended framework we have proposed also sheds new light on topological quantum computation, suggesting novel protocols based on sequential measurements described by QRFs, and bridging with TQNNs. The link between quantum simulations of complex physical systems and TQFTs has been delved in the literature, in which the adaptations of spin-networks have been contemplated. Specifically, following the inspiration that a universal representation of quantum structures exists, in which both the quantum texture of spacetime and quantum information fit, it has been argued \cite{marzuoli:02,marzuoli:05} that the dynamics associated to manipulation in quantum information can be expressed resorting to the recoupling theory of $\SU(2)$ representations. It is important to emphasize that is exactly the quantumness many authors have been referring to that allows quantum computation to be more efficient than classical computation. Indeed, clear advantages are related to some peculiar features of the former framework, including the discreetness incorporated in the representations, the natural appearance of entangled states originated by tensor product structures, and generic information encoding patterns.

Models based on $\SU(2)$ recoupling theory can be regarded as the non-Boolean generalization of the quantum circuit model \cite{marzuoli:02,marzuoli:05}. Indeed, the model involves unitary gates cast in terms of $3nj$ coefficients. These latter in turn interconnect inequivalent binary coupling schemes of $n+1$ representations of $\SU(2)$. At the same time, $j$-gates that satisfy certain algebraic identities can be realized by means of $6j$ symbols. The structure itself of the model shares similarity with partition functions emerging in TQFTs characterizing quantum gravity. This spin-network simulator, as introduced in \cite{marzuoli:05}, is a combinatorial TQFT model for computation, which can be extended to the TQNNs, as proposed in \cite{marciano:20}. A $q$-deformation of the spin-network simulator has been investigated as well \cite{kadar:08}. This naturally encodes braiding structures and enables to construct automata that efficiently perform approximate calculations of topological invariants of knots and $3$-dimensional manifolds, while sharing similarities to string-net models. The development of $q$-deformed version of the quantum simulator hinges toward the possibility to store quantum information fault-tolerantly in physical systems that support fractional statistics and whose Hilbert spaces are insensitive to local perturbations. Double Chern-Simons TQFTs have been then envisaged to represent the quantum sum over histories, along the lines of the partition function that define the Turaev-Viro model. This latter, as we have previously commented, encodes the recoupling theory of $\SU_q(2)$, which allows to identify the Hamiltonian of the Levin-Wen string-net model on the boundary with the corresponding amplitude in the Turaev-Viro model \cite{kadar:08}. 

Along this line of thought, a quantum algorithm that efficiently approximates the colored Jones polynomial was provided in \cite{garnerone:07}, within a construction that is based on the complete solution of Chern-Simons TQFTs and bridges the model to the Wess-Zumino-Witten conformal field theory. The colored Jones polynomial is calculated as the expectation value of the evolution of the $q$-deformed spin-network quantum automaton. This construction enables to realize quantum circuits simulating the automaton, hence to compute expectation values. As argued in \cite{kadar:08}, boundary $2$-dimensional lattice models can be studied within the framework of the $q$-deformed quantum spin-network simulator, and their observables can be takes into account as colored graphs that satisfy braiding relations. In \cite{kitaev:02}, $2$-dimensional quantum system that encode anyonic excitations have been deployed in order to realize a quantum computer. Unitary transformations have been used to move the excitations around one another, along the lines of topological scattering \cite{wilczek:90,lo:93,bais:98,bais:02}. A fault-tolerant computation has been then achieved by pairing excitations in pairs and observing the result of fusion. 

Another fascinating potential application of the extended framework we deepened here concerns biological information processing. Although biological information is usually assumed to be classical, it has been recently shown fully-classical models over-estimate the metabolic free-energy requirements of typical cells by orders of magnitude \cite{fields:21}, suggesting that bulk cellular biochemistry employs coherence as a computational resource.  It is, moreover, shown that single-molecule decoherence models do not rule this out.  This framework can be tested by experiments designed to detect Bell-inequality violations originated by perturbations of sister cells that have been recently-separated. The fact that both intra- and intercellular communication seems to require quantum theory suggests, in an analog model perspective, an implementation of sequential measurements by QRFs, which in turn can be recast in terms of a functorial evolution among boundary states in TQFTs. An explicit construction of canonical mammalian cortical neurons as hierarchies of QRFs has recently been proposed \cite{fgl-neurons}.  
Similar directions are pursued in \cite{smolin:20} in the quest for a ``quantum neuron'' with allusions to the Chern-Simons theory as we have considered here, along with the topology of compact 4-manifolds which may engage instantons for the YM functional. Indeed, it would be interesting to realize a formulation
of gauge potential instantons, and e.g. the gradient flow of the Chern-Simons functional in other dimensions, as these are transmigrated to action potentials and to various other topological excitations as they arise in neurophysiology; \cite{ori:22} discusses some related steps in this direction\footnote{We are grateful to Dalton Sakthivadivel for pointing out this latter reference to us.}.

\section*{Acknowledgments}
C.F. and A.M.~wish to acknowledge Filippo Fabrocini, Niels Gresnigt, Matteo Lulli and Emanuele Zappala for inspiring discussions on TQNNs.
A.M.~wishes to acknowledge support by the Shanghai Municipality, through the grant No.~KBH1512299, by Fudan University, through the grant No.~JJH1512105, the Natural Science Foundation of China, through the grant No.~11875113, and by the Department of Physics at Fudan University, through the grant No.~IDH1512092/001.

\section*{Appendices}

\subsection*{Appendix 1}

We can generalize from deterministic to probabilistic QRFs by letting the $j^{th}$ elementary QRF implement:

\begin{equation} \label{pQRF2}
\{0, 1 \}^m \xrightarrow{\{\overrightarrow{g_i} \}} ~p_{ij}~ \xrightarrow{\{\overleftarrow{g_i} \}} \{0, 1 \}^m,
\end{equation}
\noindent
where $\forall i, 0 \leq p_{ij} \leq 1$, $\sum_i p_{ij} = 1$ and the  $\overrightarrow{g_i}$ and the $\overleftarrow{g_i}$ satisfy:
\begin{equation} \label{pQRF3a}
\forall i, \exists ! e_i \in \{0, 1 \}^m, \overrightarrow{g_i} \in \{\overrightarrow{g_i} \}, ~\mathrm{and}~ \overleftarrow{g_i} \in \{\overleftarrow{g_i} \}
 ~\mathrm{such ~that}~ \overrightarrow{g_i} : e_i \mapsto p_{ij} ~\mathrm{and}~ \overleftarrow{g_i}: p_{ij} \mapsto e_i.
\end{equation}
\noindent
Here $e_i$ is a pointer value, represented as an $m$-bit string, that is obtained with probability $p_{ij}$ when the $j^{th}$ elementary QRF acts on an input $m$-bit string.  A nonelementary probabilistic QRF that acts on $m$-bit strings is an ordered sequence of $m$ elementary QRFs acting on $m$-bit strings, subject to the normalization requirement that:
\begin{equation} \label{pQRF4}
(1/m) \sum_{ij} p_{ij} = 1.
\end{equation}
\noindent
Coarse-graining to a QRF comprising $n < m$ elementary QRFs replaces $m$ with $n$ in Eq. \eqref{pQRF4}, while still summing $i$ from $1$ to $m$.

\subsection*{Appendix 2}

As in \cite{barwise:97}, classifiers in $\mathbf{Chan}$ can be appended with the structure of a \emph{local logic} $\mathcal{L}(\mathcal{A})$ relating subsets of tokens and types to each classifier. In basic terms, the local logic entails a \emph{(regular) theory} and the semantic consistency of
information flow is regulated by \emph{logic infomorphisms} between classifiers \cite[\S9.1, \S12.1]{barwise:97}, to which we refer the reader for complete details. These concepts have also been extensively reviewed in \cite{fg:19a,fg:22} (see also \cite{fg:20,fgm:21,ffgl:22}). An important ingredient of a (regular) theory is that of a \emph{sequent}, which is defined as an ``implication'' relation between subsets of types.

We recall this definition:
\begin{definition}\label{prob-1}
Two subsets $M, N \subseteq Typ(\mathcal{A})$ are related by a {\em sequent} $M \models_{\mathcal{A}} N$ if $\forall x \in Tok(\mathcal{A}), x \models_{\mathcal{A}} M \Rightarrow x \models_{\mathcal{A}} N$.
\end{definition}
\noindent
Via the sequent relation, the local logic $\mathcal{L}(\mathcal{A})$ effectively arranges the types of $\mathcal{A}$ into a hierarchy; any classifier can be constructed with its own local logic in this way by adding types as needed.  Considering all classifiers to be extended to local logics, Diagrams \eqref{ccd-1} -- \eqref{cccd-2} (and relevant constructions following) can be considered diagrams of local logics by requiring each of the infomorphisms to be a ``logic infomorphism'' that preserves the sequent structure. Logic infomorphisms are, effectively, embeddings of type hierarchies; e.g diagram \eqref{cccd-2} and can be viewed as an embedding into a ``top-level'' type hierarchy $\mathbf{C^\prime}$ that assigns an overall semantics to its inputs, followed by encodings of this top-level hierarchy into some componential representation. Loosely speaking, this and the other related diagrams may be seen as semantically encoding the `read-write' operations of the QRF.

Boolean operations can be defined on classifiers as the attendant local logic $\mathcal{L}(\mathcal{A})$ \cite[\S7.1]{barwise:97}.  To extend the type hierarchy defined by $\mathcal{L}(\mathcal{A})$ to a hierarchical Bayesian inference, we extend $\mathcal{L}(\mathcal{A})$ to a probabilistic logic by relaxing the sequent relation to require only that if $x \models_{\mathcal{A}} M$, there is some probability $P(N|M)$ that $x \models_{\mathcal{A}} N$ \cite{allwein:04}.  We can then write:
\begin{equation}\label{prob-2}
M \models_{\mathcal{A}}^{P} N =_{def} P( M \vert N)\,.
\end{equation}
Given the hierarchy in question, traversing upwards in the hierarchy imposes multiple conditioning on each `lower-level' probability distribution; traversing downwards sequentially unpacks this conditioning.  In fundamental Bayesian terms, $M$ above can be regarded as a previous event, whether observed or conjectured, and $N$ as a currently observed datum, in which case $P(M)$ becomes the prior, and $P(N)$ the evidence, that together generate a prediction. Given the likelihood $P(N \vert M)$ as the conditional obtained from weakening the sequent via Eq. \eqref{prob-2}, Bayes' theorem specifies this conditional as the posterior:
\begin{equation}\label{prob-3}
P(M \vert N) = \frac{P(N\vert M) P(M)}{P(N)}.
\end{equation}
Seen within our diagrammatic framework, a typical portion of a CCD computing a hierarchical Bayesian inference from a set of (posterior) observations $\mathcal{A}_i$ to an outcome $\mathbf{C}'$, looks like:
\begin{equation}\label{ccd-prob-1}
\begin{gathered}
\xymatrix@C=4pc{&\mathbf{C^\prime} &  \\
\mathcal{A}_1 \ar[ur]^{p_{10}(\cdot \vert \cdot)} \ar[r]_{p_{12}(\cdot \vert \cdot)}& \mathcal{A}_2 \ar[u]_{p_{20}(\cdot \vert \cdot)} \ar[r]_{p_{23}(\cdot \vert \cdot)} & \ldots ~\mathcal{A}_k \ar[ul]_{p_{k0}(\cdot \vert \cdot)}
}
\end{gathered}
\end{equation}
The above outline of ideas is the basis of the context-dependent Bayesian hierarchical structure as it is seen $\mathbf{Chan}$ (see \cite{fg:22,ffgl:22} for details).

\subsection*{Appendix 3}

The Chern character over a vector bundle $E$ is defined, in terms of the characteristic classes $\tau_k(E)$ over $E$, as
\begin{equation} \label{chern-character}
{\rm ch}(E)= {\rm rank} E + \tau_1(E) +\frac{1}{2!}  \tau_2(E) + \dots + \frac{1}{k!} \tau_k(E) +\dots\,,
\end{equation}
or alternatively, in terms of the Chern classes $c_k(E)$ over $E$, as
\begin{equation} \label{chern-character-classes}
{\rm ch}(E)= {\rm dim} E + c_1(E) +\frac{1}{2}  (c_1(E))^2 - c_2(E) +\dots\,,
\end{equation}
where dots denote terms of degree higher than four. Applying these definitions to the quaternionic line bundles $W_+$ and $W_-$, for which the first Chern classes are vanishing, one finds that ${\rm ch}(W_+)=2- c_2(W_+)$ and ${\rm ch}(W_-)=2- c_2(W_-)$. Using the fact that $W_+$ is isomorphic to its dual bundle, we can then write that
\begin{equation} \label{chern-iso}
TM\otimes \mathbb{C}= {\rm Hom}_{\mathbb{C}}(W_+, W_-) \cong W_+ \otimes W_-
\,,
\end{equation}
from which it follows the expression of the Chern class
\begin{equation} \label{chern-class-TMC}
{\rm ch}(TM\otimes \mathbb{C})= {\rm ch}(W_+) {\rm ch}(W_-) = 4 - 2 c_2(W_+) - 2 c_2(W_-) \,.
\end{equation}

Characteristic classes can be defined for quaternionic and real vector bundles. For instance, a quaternionic line bundle $E$ can be looked at as a complex vector bundle of rank $2$. Since the right multiplication by ${\bf j}$ is a conjugate-linear isomorphism from $E$ to itself, we immediately find that the first Chern class of $E$ vanishes, while the second Chern class does not, representing a relevant invariant of $E$.
Note also that a real vector bundle $E$ of rank $k$ can be complexified, giving rise to a complex vector bundle $E \times \mathbb{C}$ of rank $k$. Being the complexification isomorphic to its conjugate, its first Chern class vanishes but not the second class, which allows to define the first Pontryagin class of $E$ as $p_1(E)=-c_2(E\otimes \mathbb{C})$. Furthermore, for the first Pontryagin class it holds that
\begin{equation}
p_1(TM)= - c_2(TM \otimes \mathbb{C}) = - 2 c_2 (W_+) - 2 c_2 (W_-)= -2 c_2(W)\,.
\end{equation}
As a consequence, if $E$ is a $\SU(2)$-bundle, then one finds
$(2-c_2(E)+\dots)^2=4-c_2(E \otimes E)+\dots $, which implies $c_2(E)=\frac{1}{4}c_2(E \otimes E)$.

Finally, the $\hat{A}$-polynomial is defined by the power series
\begin{equation}
\hat{A}(TM)= 1-\frac{1}{24} p_1(TM) + \dots \,,
\end{equation}
dots denoting terms in higher order Pontryagin classes, which are all vanishing in dimension $d\leq 4$.

\subsection*{Appendix 4}

Being more specific about the construction of gauge-networks and their relation to noncommutative geometry we referred to in \S \ref{quivers-gauge-networks}, we introduce a category $\mathcal{C}_0$ that comprises finite spectral triples, the morphisms of which are pairs of an algebra morphism and a unitary operator provided with a compatibility condition. A subcategory then exists that includes finite spectral triples with trivial Dirac operator. We may start introducing, for finite-dimensional algebra representations and finite spectral triples, the category of finite-dimensional algebras, carrying representation on Hilbert spaces.

\begin{definition} \label{category-finite-dimensional-algebras-1}
The category $\mathcal{C}_0$ has as objects the triple  $(\mathcal{A}, \lambda, \mathcal{H})$, with $\mathcal{A}$ denoting a finite-dimensional $C*$-algebra and $\lambda$ denoting $\star$-representation on an inner-product space $\mathcal{H}$.  A morphism in Hom$((\mathcal{A}_1, \mathcal{H}_1),(\mathcal{A}_2, \mathcal{H}_2))$ can be realized by recovering a map $(\phi,L)$, which comprises a unital $\star$-algebra map $\phi:\mathcal{A}_1 \rightarrow \mathcal{A}_2$ and a unitary map $L:\mathcal{M}_1 \rightarrow \mathcal{M}_2$ such that, for each $A\in \mathcal{A}_1$, $L \lambda_1(A) L^*=\lambda_2(\phi(A))$.
\end{definition}

$\mathcal{C}_0$ could have been also introduced as the category of finite spectral triples $(\mathcal{A}, \mathcal{H} , D)$, the Dirac operator of which is vanishing. The extension of the category $\mathcal{C}$, of which $\mathcal{C}_0$ is a subcategory, in order to account for the Dirac operator immediately follows:

\begin{definition} \label{category-finite-dimensional-algebras-2}
The objects of the category $\mathcal{C}$ are the finite spectral triples $(\mathcal{A}, \lambda, \mathcal{H}, D)$, with $\mathcal{A}$ denoting a finite-dimensional $C*$-algebra, $\lambda$ denoting $\star$-representation on an inner-product space $\mathcal{H}$, and $D$ denoting a symmetric linear operator on $\mathcal{H}$.  A morphism in Hom$((\mathcal{A}_1, \mathcal{H}_1, D_1),(\mathcal{A}_2, \mathcal{H}_2, D_2))$ can be realized by recovering a map $(\phi,L)$, which comprises a unital $\star$-algebra map $\phi:\mathcal{A}_1 \rightarrow \mathcal{A}_2$ and a unitary map $L:\mathcal{M}_1 \rightarrow \mathcal{M}_2$ such that, for each $A\in \mathcal{A}_1$, it holds that $L \lambda_1(A) L^*=\lambda_2(\phi(A))$, and $L D_1L^*=D_2$ is also fulfilled.
\end{definition}

If we now introduce oriented graphs $\Gamma$, and regard these lettere as quivers, we may describe gauge theories representing quivers in $\mathcal{C}$.

\begin{definition} \label{quiver-generalized}
A representation $\pi$ of a quiver $\Gamma$ in a category corresponds to an association of objects $\pi_v$ in that category to each vertex $v \in V$ of $\Gamma$, and morphism $\pi_e$ in Hom$(\pi_{s(e)},\pi_{t(e)})$ to each oriented edge $e\in E$ of $\Gamma$.
Two representations $\pi$ and $\pi'$ of a quiver $\Gamma$ that are in the same category are equivalent if $\pi_v=\pi_v'$ for every $v\in V$, and there exists a family of invertible morphisms $\phi_v\in$Hom$(\pi(v),\pi(v))$ labelled by $v$ that are such that $\pi_e=\phi_{t(e)}\circ \pi_e' \circ \phi_{s(e)}^{-1}$.
\end{definition}

If we regard a quiver $\Gamma$ as a category, a representation of $\Gamma$ is provided by a functor $\pi:\Gamma \rightarrow \mathcal{C}$. If a category $\mathcal{C}$ ($\mathcal{C}_0$) is considered, a representation $\pi$ of the quiver $\Gamma$ assigns spectral triples $(\mathcal{A}_v,\mathcal{H}_v,D_v)$ ($D_v=0$ for $\mathcal{C}_0$) to each vertex $v\in V$ of $\Gamma$ and pairs $(\phi, L)\in$Hom$((\mathcal{A}_{s(e)}, \mathcal{H}_{s(e)}, D_{s(e)}),(\mathcal{A}_{t(e)}, \mathcal{H}_{t(e)}, D_{t(e)}))$ to each edge $e\in E$ of $\Gamma$. The space of the functorial quiver representations $\pi:\Gamma  \rightarrow \mathcal{C}$ can be denoted with $\mathcal{X}$. The group $\mathcal{G}$ is the collection of invertible morphisms $(\phi_v,\psi_v)$ for each vertex $v\in V$ of $\Gamma$. The form of the space $\mathcal{X}$ and the quotient $\mathcal{X}/\mathcal{G}$ has been explicitly determined in Ref.~\cite{marcolli:14}.

The Hilbert space of a quantum theory can be constructed on the classical configuration space $\mathcal{X}/\mathcal{G}$, which is a measure space equipped with products and sums of the Haar measures on the unitary groups. This enables to consider $L^2(\mathcal{X})$. The Hilbert space hence introduced is endowed with an action of $\mathcal{G}$, which is induced by the action of $\mathcal{G}$ on $\mathcal{H}$. An explicit description of the space $L^2(\mathcal{X}/\mathcal{G})\simeq L^2(\mathcal{X})^\mathcal{G}$ of functions that are $\mathcal{G}$-invariant on $\mathcal{X}$ can be provided.  The following theorem can be proved \cite{marcolli:14}:

\begin{theorem} \label{gauge-thm2}
For $G$ a compact Lie group, it can be proved that the unitary map $L$ introduces an isomorphism of $G \times G$ representations
\begin{equation}
L^2(G)\simeq \bigotimes_{\rho \in \hat{G}} \rho \otimes \rho^*\,,
\nonumber
\end{equation}
with element $(g_1,g_2)\in G \times G$ acting according to the rules: for $f\in L^2(G)$, it holds that $((g_1,g_2)f)(x)=f(g_1^{-1}x g_2)$; for $y_1\in \rho$ and $y_2\in \rho^*$, it holds that $(g_1,g_2)(y_1 \otimes y_2) = \rho(g_1)y_1 \otimes \rho^*(g_2)y_2$.
\end{theorem}

Within the configuration space of quiver representations in the category of finite spectral triples, gauge-networks can be introduced as an orthonormal basis of a corresponding Hilbert space.

\begin{definition} \label{gauge-network}
A gauge-network is provided by the data $\{ \Gamma, (\mathcal{A}_v, \lambda_v, \mathcal{H}_v, \iota_v)_v, (\rho_e, \mathbb{B}_e)_e\}$, with: 1) $\Gamma$ is an oriented graph; 2) $(\mathcal{A}_v, \lambda_v, \mathcal{H}_v)$ is an object in the category $\mathcal{C}_0$ for each $v\in V$ of $\Gamma$; 3) for each $e\in E$ of $\Gamma$, $\rho_e$ is a representation of the group $G_e$; 4) for each $e\in E$ of $\Gamma$, $\mathbb{B}_e$ is a Bratelli diagram for $\star$-algebra maps of the type $\mathcal{A}_{s(e)}\rightarrow \mathcal{A}_{t(e)}$, with subdiagrams $\tilde{\mathbb{B}}$ for $\tilde{\mathcal{A}}_{s(e)}\rightarrow \tilde{\mathcal{A}}_{t(e)}$ and $\mathbb{B}_0$ for ${\mathcal{A}}_{s(e)}\rightarrow {\rm ker}\lambda_{t(e)}$; 5) for each vertex $v\in V$, the intertwiners of the group $\mathcal{G}_v \simeq U(\mathcal{A}_v) >\!\!\!\lhd \,S(\mathcal{A}_v)$ exist such that
$$
\iota_v: \rho_{e_1'} \otimes \cdots  \otimes \rho_{e_k'} \rightarrow \rho_{e_1} \circ \phi_\mathbb{B} \otimes \cdots  \otimes \rho_{e_k}  \circ \phi_\mathbb{B}  \,,
$$
with $e_1', \dots, e_k'$ edges incoming to $v$, and $e_1, \dots, e_k$ edges outgoing from $v$, and the isotropy group $K_{\mathbb{B}_e}=\mathcal{U}({\rm ker} \, \lambda_{t(e)})_{\mathbb{B}_{e0}}$.
\end{definition}

Given Definition \ref{gauge-network}, abelian spin-networks and ${\rm U}(N)$ spin-networks can be introduced as examples --- see e.g. \cite{marcolli:14}.

A notion of evolution can be introduced for gauge-networks, once an Hamiltonian operator $\mathbb{H}$ on $L^2(\mathcal{X})$ is introduced as the sum of invariant Laplacians on homogeneous spaces. Indeed, since the Casimir operators of the Lie groups $\mathcal{U}(\mathcal{A}_{t(e)})$ are $\mathcal{U}(\mathcal{A}_{t(e)})$-biinvariant, on each $L^2$ we can introduce Laplacian operators $C^2_{\mathcal{U}(\mathcal{A}_{t(e)})}$. The following proposition has been then proved in Ref.~\cite{marcolli:14}:

\begin{proposition} \label{evolution}
Once a pair $\{ \mathcal{A}_v, \mathcal{H}_v \}$ has been picked up at each vertex $v\in V$ of $\Gamma$ and Bratelli diagrams $\mathbb{B}_e$ at each edge $e\in E$ of $\Gamma$ have been selected, the tensor product or quadratic Casimirs $C^2_{\mathcal{U}(\mathcal{A}_{t(e)})}$
\begin{equation}
\sum_e \left( 1\!\!1 \otimes \cdots \otimes C^2_{\mathcal{U}(\mathcal{A}_{t(e)})} \otimes \cdots \otimes 1\!\!1 \right)
\,,
\end{equation}
is a self-adjoint operator on the finite tensor product $\bigotimes_e L^2(\mathcal{U}(\mathcal{A}_{t(e)})/\mathcal{U}(\lambda_{t(e)})_{ \mathbb{B}_{e0} })$.

The sum of the operators hence defined is a self-adjoint operator $\mathbb{H}$ on
$$L^2(\mathcal{X})\simeq \bigoplus_{\{ \mathcal{A}_v, \mathcal{H}_v \}} \bigoplus_{\mathbb{B}_e} \bigotimes_e L^2(\mathcal{U}(\mathcal{A}_{t(e)})/ \mathcal{U}(\lambda_{t(e)})_{\mathbb{B}_{e0}} )\,.$$
\end{proposition}

Graphs $\Gamma$, here considered as quivers, can be embedded in a Riemannian spin-manifolds $\mathcal{M}$, over which gauge theories are defined. This gives rise to a twisted Dirac operator on $\Gamma$. It has been shown in \cite{marcolli:14} that the corresponding spectral action reduces to the Wilson action of lattice gauge theory.

\subsection*{Appendix 5}

Quiver theories supported on a $2$-dimensional toroidal lattice have been considered in \cite{razamat:21} for a $\mathcal{N}=1$ supersymmetric theory with in $d=4$. The gauge group is $SU(N)^{L_1\times L_2}$, with $L_1$ nodes of the lattice along the vertical axis and $L_2$ nodes of the lattice along the horizontal axis, while the matter content comprises chiral superfields $\mathfrak{U}_p$, $\mathfrak{B}_p$ and $\mathfrak{C}_p$ that transform in the fundamental representation of one of the $L_1\times L_2$ $\SU(N)$, and are associated to three types of edges, one vertical and two diagonals, forming triangular plaquettes $p$ in the lattice. The interaction is expressed in terms of $\mathcal{N}=1$ $\SU(N)$ supersymmetric gauge fields associated to each vertex, and involves fields that end or are originated at that vertex. The superpotential interaction $W$ can be derived summing over the contributions localized at each triangular plaquette,
\begin{equation}
W=\sum_p \lambda_p {\rm Tr} \, \mathfrak{U}_p\, \mathfrak{B}_p\,\mathfrak{C}_p\,,
\end{equation}
with $\lambda_p$ coupling constants. Quiver theories of this type are conformal \cite{seiberg:21} in the UV. Combinations of couplings can be recovered from the conformally invariant theory following the renormalization group flow \cite{vijay:15,nandkishore:20}. Additional global symmetries are also present, with dimension of the symmetry group related to the dimension of the lattice. There are also subsystem symmetries of the lattice that can be realized, and are responsible for the gluing of the lattice into a torus.

Gauge-invariant excitations on the lattice can be considered, which are local in that their operators are expressed in terms of fields localized on a single edge of the lattice. Chiral baryons constructed from the edge of the lattice provide such a set of operators of the type
\begin{eqnarray}
A_{ij}
&=&
\epsilon^{a_1\dots a_N} \, \epsilon_{b_1\dots b_N} \prod_{l=1}^N \, (\mathfrak{U}_{ij})^{b_l}_{a_l} \,,
 \\
B_{ij}
&=&
\epsilon_{a_1\dots a_N} \, \epsilon^{b_1\dots b_N} \prod_{l=1}^N \, (\mathfrak{B}_{ij})^{a_l}_{b_l}  \,,
\\
C_{ij}
&=&
\epsilon_{a_1\dots a_N} \, \epsilon^{b_1\dots b_N} \prod_{l=1}^N \, (\mathfrak{C}_{ij})^{a_l}_{b_l}  \,,
\end{eqnarray}
where gauge indices are contracted and fields and baryons are labelled by their colouring symmetries. It has been shown in Ref.~\cite{razamat:21}, by inspecting their correlation functions, that baryonic excitations are immobile on the lattice, while pairs of $B$ and $C$ baryon excitations (as well as $B$ and $A$ or $C$ and $A$) are in principle movable along the quiver lattice directions, a similar result holding also for triplet of $A$, $B$ and $C$ baryons.

\end{document}